\pdfoutput=1
\documentclass[10pt]{article}

\usepackage{setspace}
\usepackage{fullpage}
\usepackage[numbers]{natbib}  
\usepackage{amssymb}
\usepackage{amsmath}
\usepackage{rotating}
\usepackage{appendix}
\usepackage{adjustbox}
\usepackage{comment}
\usepackage{amsthm}
\usepackage{booktabs}
\usepackage{hyperref}

\theoremstyle{definition}

\newtheorem{proposition}{Proposition}

\newtheorem{customtheorem}{Theorem}

\begin{document}

\title{Wasserstein Robust Market Making via Entropy Regularization}

\author{
    Zhou Fang \\ Department of Mathematics, The University of Texas at Austin 
    \and
    Arie Israel \\ Department of Mathematics, The University of Texas at Austin
}

\maketitle

\begin{abstract}
    In this paper, we introduce a robust market making framework based on Wasserstein distance, utilizing a stochastic policy approach enhanced by entropy regularization. We demonstrate that, under mild assumptions, the robust market making problem can be reformulated as a convex optimization question. Additionally, we outline a methodology for selecting the optimal radius of the Wasserstein ball, further refining our framework's effectiveness.
\end{abstract}

\section{Introduction}
Market-making is a popular trading activity typically undertaken by proprietary trading firms. A market maker sets both bid and ask prices for a particular financial asset, capitalizing on the small spreads for profit. Because each transaction yields only a modest profit, they must engage in frequent trading, linking market-making closely with high-frequency trading strategies. Additionally, market makers are tasked with managing their inventory effectively to mitigate risks associated with sudden market fluctuations. 

The study of market-making was initially pioneered by \cite{ho1981optimal}, but it lay dormant for decades until it was revived by the paper \cite{avellaneda2008high}. Following these seminal works, a wealth of subsequent research has emerged, spanning various fields. This includes option market making, as discussed in \cite{baldacci2021algorithmic} and \cite{stoikov2009option}, FX exchange market-making \cite{barzykin2021market}, and market making in cryptocurrencies \cite{cartea2023decentralised}. Further explorations consider market making under specific constraints or conditions \cite{cartea2019market} and \cite{cartea2020market}. The literature on market-making within the field of financial mathematics is vast and varied. Although each paper offers a unique perspective, collectively, they address market-making through the lens of stochastic control problems, assuming the underlying dynamics are known.

However, in reality, it is often the case that accurate information on the underlying dynamics, such as returns and volatility, is not available. There are few studies that incorporate the uncertainty of these underlying dynamics, likely due to the complexity of addressing such uncertainties. Recent developments in Wasserstein Distributionally Robust Optimization (DRO) offer a promising tool for addressing these challenges. In this project, we employ techniques from Wasserstein DRO to tackle the market-making problem under uncertain dynamics.

The seminal work by Kuhn et al. \cite{kuhn2019wasserstein} has played a pivotal role in bridging the gap between theoretical foundations and practical applications of the Wasserstein distributionally robust optimization (DRO). This celebrated paper not only offers a comprehensive tutorial on the subject but also intricately details the relationship between Wasserstein DRO and various machine learning paradigms, setting the stage for a multitude of research directions. Since its publication, there has been a surge in scholarly interest towards exploring the depth and breadth of Wasserstein DRO, evidenced by several influential studies. Notably, a series of papers by Blanchet et al., including \cite{blanchet2019robust}, \cite{blanchet2022confidence}, and \cite{blanchet2021statistical}, alongside Mohajerin and Kuhn \cite{mohajerin2018data}, have conducted rigorous investigations into data-driven approaches within the framework of Wasserstein DRO. These works collectively advance our understanding of how data can inform more robust and reliable optimization strategies, particularly in the context of uncertain environments.

Moreover, the evolution of this research area has seen significant contributions aimed at generalizing and extending the strong duality results inherent in Wasserstein DRO. Papers such as those by Gao et al. \cite{gao2022wasserstein}, \cite{gao2023distributionally}, and Yang and Gao \cite{yang2022wasserstein}, have been instrumental in elucidating more tractable and scalable approaches to implementing Wasserstein DRO in larger, more complex systems. Their work expands the theoretical underpinnings of strong duality and provides a clearer path towards practical implementation, enhancing the applicability of Wasserstein DRO across various domains, including finance, \cite{blanchet2022distributionally}

\section{Model Setup}

In the market-making problem, a key aspect is modeling the arrival process of market orders. This is often represented as a Poisson process, where the arrival intensity reflects the mean arrival rate. Following traditional market-making literature, the intensities of market order arrivals are influenced by the bid and ask prices and their proximity to the mid-price. Specifically, the closer the bid (ask) price is to the mid-price, the higher the likelihood of a market sell (buy) order being executed.

We first discuss several observed numerical phenomena in this context and then outline our specific assumptions regarding market order arrival intensities. We aim to set assumptions on the $\mathbb{R}$-valued random variables \( \Delta N^+(\epsilon^+) \) and \( \Delta N^-(\epsilon^-) \), where \( \Delta N^+(\epsilon^+) \) represents the number of market buy orders given the current ask spread \( \epsilon^+ \), and \( \Delta N^-(\epsilon^-) \) represents the number of market sell orders given the current bid spread $\epsilon^-$ for the next period of time $\Delta t$. Here, we assume $\epsilon^+ > 0$ and $\epsilon^- > 0$.

We assume the random variables \( \Delta N^+(\epsilon^+) \) and \( \Delta N^-(\epsilon^-) \) are distributed according to probability distributions \( \mathbb{Q}^+(\epsilon^+) \) and \( \mathbb{Q}^-(\epsilon^-) \) on $\mathbb{R}$, respectively. These parametric families of probability distributions capture the uncertainty in market order arrivals for different spreads. When the dependence on the spreads is understood from context, we write the distributions as $\mathbb{Q}^+$ and $\mathbb{Q}^-$. We will denote the spread variables as $\boldsymbol{\epsilon} = (\epsilon^+,\epsilon^-)$.

The initial formulation considers the following optimization problem:
\begin{align*}
    & \max_{\epsilon^+, \epsilon^-} \quad  \mathbb{E}^{\mathbb{Q}^+ \otimes \mathbb{Q}^-} \Bigg[ (S + \epsilon^+) \Delta N^+(\epsilon^+) - (S - \epsilon^-) \Delta N^-(\epsilon^-) - \gamma \left( Q + \Delta Q(\boldsymbol{\epsilon}) \right)^2 \Bigg], \\
    & \qquad \text{where} \quad \Delta Q(\boldsymbol{\epsilon}) = \Delta N^+(\epsilon^+) - \Delta N^-(\epsilon^-).
\end{align*}

In the above:

\begin{itemize}
    \item \( S \) is the asset's mid-price.
    \item \( Q \) is the current inventory of the asset.
    \item \( \epsilon^+ \) and \( \epsilon^- \) are the ask and bid spreads, respectively.
    \item \( \Delta N^+(\epsilon^+) \sim \mathbb{Q}^+(\epsilon^+) \) and \( \Delta N^-(\epsilon^-) \sim \mathbb{Q}^-(\epsilon^-) \).
    \item \( \gamma \geq 0 \) is the penalty parameter for holding inventory in the next period.
\end{itemize}
In the objective function, the term 
\[
(S + \epsilon^+) \Delta N^+(\epsilon^+) - (S - \epsilon^-) \Delta N^-(\epsilon^-)
\] 
represents the change in cash. Additionally, since \( S \) is the current asset price, we assume that in the next period (after a time interval \( \Delta t \)), the asset price becomes \( S + \Delta S \), where \( \Delta S \) is the random change in the asset price. Here, \( \Delta S \) follows a probability measure that is independent of the measures of $\mathbb{Q}^+(\epsilon^+)$ and $\mathbb{Q}^-(\epsilon^-)$.

If we assume that the expected change in the asset price is zero, then it follows that the expected change in cash is equivalent to the expected change in total wealth. Finally, the term 
\[
-\gamma \left( Q + \Delta Q \right)^2
\] 
penalizes holding inventory in the next period, discouraging excessive accumulation or depletion of assets.

In the above classical setting, the policy is deterministic, meaning that given the current state \( (S, Q) \), the policy prescribes a deterministic bid and ask spread \(  \epsilon^+(S, Q) \) and \( \epsilon^-(S, Q) \), respectively. However, assuming a deterministic policy not only limits our choice of possible policies but also poses difficulties when modifying the market-making problem into a robust optimization setting. Therefore, we are led to introduce the market-making problem which incorporates a stochastic policy. A \emph{stochastic policy} is a probability distribution over possible bid and ask spreads, determined by the current state. Specifically, we define a stochastic policy \( \boldsymbol{\pi}(\boldsymbol{\epsilon} | S, Q) \) to be a probability density function (pdf) for the variables \( \boldsymbol{\epsilon}=(\epsilon^+,\epsilon^-) \in (\mathbb{R}^{+})^2 \) given the state \( (S, Q) \). The policy \( \boldsymbol{\pi}(\boldsymbol{\epsilon} | S, Q) \) determines the probability of choosing specific spreads \( \boldsymbol{\epsilon}= (\epsilon^+, \epsilon^-) \) given the state \( (S, Q) \). For simplicity, we write $\boldsymbol{\pi}(\boldsymbol{\epsilon}) = \boldsymbol{\pi}(\boldsymbol{\epsilon} | S, Q) $. Thus, $\int_{(\mathbb{R}^+)^2} \boldsymbol{\pi}(\boldsymbol{\epsilon}) d \boldsymbol{\epsilon} = 1$ and $\boldsymbol{\pi}(\boldsymbol{\epsilon}) \geq 0$.

The stochastic market-making problem with entropy regularization is then represented as
\begin{align}
    & \max_{\boldsymbol{\pi}} \;\; \biggl\{\int_{(\mathbb{R}^+)^2} \boldsymbol{\pi}(\boldsymbol{\epsilon}) \, \mathbb{E}^{\mathbb{Q}^+ \otimes \mathbb{Q}^-} \Big[ (S + \epsilon^+) \Delta N^+(\epsilon^+) - (S - \epsilon^-) \Delta N^-(\epsilon^-) - \gamma \left( Q + \Delta Q(\boldsymbol{\epsilon}) \right)^2 \Big] \, d\boldsymbol{\epsilon} \nonumber \\
    & \qquad - \eta \int_{(\mathbb{R}^+)^2} \boldsymbol{\pi}(\boldsymbol{\epsilon} ) \log \boldsymbol{\pi}(\boldsymbol{\epsilon}) \, d\boldsymbol{\epsilon} \biggr\},  \qquad \text{where} \quad \Delta Q(\boldsymbol{\epsilon}) = \Delta N^+(\epsilon^+) - \Delta N^-(\epsilon^-). \label{orig_problem}
\end{align}
Here, \( \eta > 0 \) is an entropy regularization parameter that encourages exploration over the spread settings, preventing the policy from over-committing to specific values of \( \epsilon^+ \) and \( \epsilon^- \). The second integral represents the entropy of the probability distribution $\boldsymbol{\pi}$, and, as usual, we interpret $0 \log(0) = 0$.

Another of the major modifications that will be done is making the above loss function into its robust version. The reason is that it will allow us to treat problems where we do not know the exact distribution $\mathbb{Q}^\pm(\epsilon^\pm)$ of the random variable $\Delta N^\pm(\epsilon^\pm)$, but only have access to estimates for these distributions. We propose the following two assumptions that will be used in this paper. 

\begin{itemize}
    \item \textbf{Assumption 1 (Model Uncertainty): } $\mathbb{Q}^\pm (\epsilon^\pm)$ is the distribution of market buy/sell orders if the proposed bid/ask spread is $\epsilon^\pm$. The assumption in this paper is that $\mathbb{Q}^\pm (\epsilon^\pm)$ is transformed from an unknown meta-distribution $\mathbb{Q}_0^\pm$, where the transformation is shifted by mean $f^\pm(\epsilon^\pm)$ and scaled by standard deviation $h^\pm(\epsilon^\pm)$, as in equation \eqref{eqn:scaled_dist}. The transformations $f^\pm(\epsilon^\pm)$ and $h^\pm(\epsilon^\pm)$, are continuous functions of $\epsilon^\pm$, which are known in advance. We assume that the meta-distributions $\mathbb{Q}_0^\pm$ have zero mean. We do not know the distributions $\mathbb{Q}_0^\pm$ in advance, but we assume that we have access to a sequence of independent samples of these distributions.
\end{itemize}

We let $\Delta\widetilde{N}^\pm$ denote two independent copies of a scalar random variable distributed according to the distribution $\mathbb{Q}_0^\pm$. Then, according to Assumption 1, we shall assume:
\begin{align}\label{eqn:scaled_dist}
    \Delta N^{\pm}(\epsilon^\pm) = h^\pm(\epsilon^\pm)\Delta \widetilde{N}^\pm + f^\pm(\epsilon^\pm).
\end{align}

Under the Wasserstein DRO framework, the next step is to define the uncertainty sets. First, we will need a second assumption on the empirical data

\begin{itemize}
    \item \textbf{
        Assumption 2 (Empirical Data)
    }
    Assume that the distributions of standardized buy and sell order arrivals, denoted by \(\mathbb{Q}_0^+\) and \(\mathbb{Q}_0^-\), are independent but may deviate from their respective empirical distributions. The empirical data for standardized buy and sell orders are given by \(\{\Delta \widetilde{N}_i^+\}_{i=1}^n\) and \(\{\Delta \widetilde{N}_i^-\}_{i=1}^n\), respectively. Notice the empirical data is derived from historical observations, it naturally forms a time series. As time progresses, more data becomes available.  Their corresponding empirical distributions denoted as \(\widehat{\mathbb{Q}}_n^+\) and \(\widehat{\mathbb{Q}}_n^-\). More specifically, the empirical distributions are  

\[
\widehat{\mathbb{Q}}_n^+ = \frac{1}{n} \sum_{i=1}^{n} \delta_{\Delta \widetilde{N}_i^+}, \quad
\widehat{\mathbb{Q}}_n^- = \frac{1}{n} \sum_{i=1}^{n} \delta_{\Delta \widetilde{N}_i^-},
\]
where \(\delta_x\) denotes the Dirac measure centered at \(x\)
\end{itemize}

To capture this uncertainty, each distribution is assigned its own uncertainty set around the empirical distributions, defined as:
\[
\mathcal{U}^+_{n,\delta} = \left\{ \widetilde{\mathbb{Q}}^+ : W_2\left( \widetilde{\mathbb{Q}}^+, \widehat{\mathbb{Q}}_n^+ \right) \leq \delta \right\} \quad \text{and} \quad \mathcal{U}^-_{n,\delta} = \left\{ \widetilde{\mathbb{Q}}^- : W_2\left( \widetilde{\mathbb{Q}}^-, \widehat{\mathbb{Q}}_n^- \right) \leq \delta \right\},
\]
where \(W_2\) represents the Wasserstein-2 distance with the Euclidean cost function.

Since \(\Delta \widetilde{N}^+\) and \(\Delta \widetilde{N}^-\) are assumed to be independent, the combined uncertainty set is represented by the product distribution \(\widetilde{\mathbb{Q}}^+ \otimes \widetilde{\mathbb{Q}}^-\), each independently varying within its respective Wasserstein ball. This structure enables a decision-independent Distributionally Robust Optimization (DRO) framework that accommodates uncertainty in both buy and sell order arrivals.

Now, we are able to discuss the robust version of the optimization problem by considering that the policy needs to be optimal under the worst-case scenario. By unfolding the market order arrival process and substituting the standardized market order arrivals, we modify the optimization problem as follows:

\begin{align}
    &\max_{\boldsymbol{\pi}} \inf_{\widetilde{\mathbb{Q}}^+ \in \mathcal{U}^+_{n, \delta}, \, \widetilde{\mathbb{Q}}^- \in \mathcal{U}^-_{n, \delta}} \, \int_{(\mathbb{R}^+)^2} \boldsymbol{\pi}(\epsilon^+, \epsilon^-) \mathbb{E}^{\widetilde{\mathbb{Q}}^+ \otimes \widetilde{\mathbb{Q}}^-}   \bigg[ (S + \epsilon^+) \Big( h^{+}(\epsilon^+) \Delta \widetilde{N}^+ + f^{+}(\epsilon^+) \Big) \nonumber \\
    - & (S - \epsilon^-) \Big( h^{-}(\epsilon^-) \Delta \widetilde{N}^- + f^{-}(\epsilon^-) \Big) - \eta \Big( Q + f^{+}(\epsilon^+) - f^{-}(\epsilon^-) + h^{+}(\epsilon^+) \Delta \widetilde{N}^+ - h^{-}(\epsilon^-) \Delta \widetilde{N}^- \Big)^2 \bigg] d\epsilon^+ d\epsilon^- \nonumber \\
    - & \gamma \int_{(\mathbb{R}^+)^2} \boldsymbol{\pi}(\epsilon^+, \epsilon^-) \log \boldsymbol{\pi}(\epsilon^+, \epsilon^-) \, d\epsilon^+ d\epsilon^-. \label{orig_problem_2}
\end{align}
Here, as before, we regard the state $(S,Q)$ as fixed, $\eta,\gamma > 0$ are model parameters, and $\Delta \widetilde{N}^+$ and $\Delta \widetilde{N}^-$ are distributed according to the probability measures $\widetilde{\mathbb{Q}}^+$ and $\widetilde{\mathbb{Q}}^-$. For the size of the Wasserstein ball $\delta$, we will introduce a way for picking the optimal in the last section, which involves definition of Wasserstein robust profile.

To simplify notation, we define:
\begin{align}
    A &= (S + \epsilon^+) h^{+}(\epsilon^+) \nonumber \\
    B &= (S - \epsilon^-) h^{-}(\epsilon^-) \nonumber \\
    C &= Q + f^{+}(\epsilon^+) - f^{-}(\epsilon^-)
\end{align}
With these definitions, the optimization problem simplifies to:
\begin{align}
    \max_{\boldsymbol{\pi}} \quad & \inf_{\widetilde{\mathbb{Q}}^+ \in \mathcal{U}^+_{n, \delta}, \, \widetilde{\mathbb{Q}}^- \in \mathcal{U}^-_{n, \delta}} \, \int_{(\mathbb{R}^+)^2} \boldsymbol{\pi}(\epsilon^+, \epsilon^-) \, \mathbb{E}^{\widetilde{\mathbb{Q}}^+ \otimes \widetilde{\mathbb{Q}}^-} \Bigg[ A \Delta \widetilde{N}^+ - B \Delta \widetilde{N}^- \nonumber \\
    & \quad - \eta \left( C + h^{+}(\epsilon^+) \Delta \widetilde{N}^+ - h^{-}(\epsilon^-) \Delta \widetilde{N}^- \right)^2 + (S + \epsilon^+) f^+(\epsilon^+) 
                - (S - \epsilon^-) f^-(\epsilon^-) \Bigg] \, d\epsilon^+ d\epsilon^- \nonumber \\
    & \quad - \gamma \int_{(\mathbb{R}^+)^2} \boldsymbol{\pi}(\epsilon^+, \epsilon^-) \log \boldsymbol{\pi}(\epsilon^+, \epsilon^-) \, d\epsilon^+ d\epsilon^- \label{eqn:main_opt}
\end{align}

\section{Main Results}
In this section, we present our main result: the aforementioned optimization problem, denoted as \eqref{eqn:main_opt} has optimal policy stated in theorem \ref{main_thm}

\begin{customtheorem}\label{main_thm}
    The optimal policy for the problem stated in \eqref{eqn:main_opt} is given by:
    \begin{align}
        \boldsymbol{\pi}^*(\epsilon_0^+, \epsilon_0^-) 
        &=
        \frac{M^*(\epsilon_0^+, \epsilon_0^-)}{\int_{(\mathbb{R}^+)^2} M^*(\epsilon^+, \epsilon^-) \, d\epsilon^+ \, d\epsilon^-}, \nonumber \\ 
        M^*(\epsilon^+, \epsilon^-) 
        &= 
        \exp \Bigg\{ 
            \frac{1}{\gamma} \Bigg[ 
                (A - 2 \eta C h^+(\epsilon^+)) \alpha^{*,+} 
                - (B - 2 \eta C h^-(\epsilon^-)) \alpha^{*,-} \nonumber \\
        & - \eta \Big( h^+(\epsilon^+)^2 \beta^{*,+} - 2 h^+(\epsilon^+) h^-(\epsilon^-) \alpha^{*,+} \alpha^{*,-} 
                + h^-(\epsilon^-)^2 \beta^{*,-} \Big) 
            \Bigg] 
        \Bigg\} L(\epsilon^+, \epsilon^-), \nonumber \\ 
        L(\epsilon^+, \epsilon^-) 
        &= 
        \exp \Bigg\{ 
            \frac{1}{\gamma} \Big[ 
                (S + \epsilon^+) f^+(\epsilon^+) 
                - (S - \epsilon^-) f^-(\epsilon^-) 
                - \eta C^2 
            \Big] 
        \Bigg\}.
    \end{align}
    
    The parameters \(\alpha^{*,+}\), \(\alpha^{*,-}\), \(\beta^{*,+}\), and \(\beta^{*,-}\) in this policy are the optimal solutions of the two-dimensional optimization problem:
    \begin{align}
        &\sup_{\widetilde{\mathbb{Q}}^+ \in \mathcal{U}^+_{n, \delta}, \, \widetilde{\mathbb{Q}}^- \in \mathcal{U}^-_{n, \delta}} -\gamma \int_{(\mathbb{R}^+)^2} M(\epsilon^+, \epsilon^-) \, d\epsilon^+ \, d\epsilon^- \nonumber \\
        \text{subject to} \nonumber \\
        &\quad \alpha^\pm = \mathbb{E}^{\widetilde{\mathbb{Q}}^\pm} [\Delta \widetilde{N}^\pm] \in \left[\alpha_n^\pm - \sqrt{\delta}, \, \alpha_n^\pm + \sqrt{\delta}\right], \nonumber \\
        &\quad \beta^\pm = \mathbb{E}^{\widetilde{\mathbb{Q}}^\pm} [(\Delta \widetilde{N}^\pm)^2] = \Big( \sqrt{\beta_n^\pm - (\alpha_n^\pm)^2} + \sqrt{\delta - (\alpha^\pm - \alpha_n^\pm)^2} \Big)^2 + (\alpha^\pm)^2. \nonumber
    \end{align}
    where $\alpha_n^\pm = \mathbb{E}^{\widehat{\mathbb{Q}}^\pm_n} [\xi^\pm]$ and $\beta_n^\pm = \mathbb{E}^{\widehat{\mathbb{Q}}^\pm_n} [(\xi^\pm)^2]$, with $\xi^\pm \sim \widehat{\mathbb{Q}}^\pm_n $, and
    \begin{align}
        M(\epsilon^+, \epsilon^-) 
        &= 
        \exp \Bigg\{ 
            \frac{1}{\gamma} \mathbb{E}^{\widetilde{\mathbb{Q}}^+ \otimes \widetilde{\mathbb{Q}}^-} \Big[ 
                (A - 2 \eta C h^+(\epsilon^+)) \Delta \widetilde{N}^+ 
                - (B - 2 \eta C h^-(\epsilon^-)) \Delta \widetilde{N}^- \nonumber \\
        & \quad - \eta \Big( h^+(\epsilon^+) \Delta \widetilde{N}^+ 
                - h^-(\epsilon^-) \Delta \widetilde{N}^- \Big)^2 
            \Big] 
        \Bigg\} L(\epsilon^+, \epsilon^-) \nonumber \\ 
        L(\epsilon^+, \epsilon^-) 
        &= 
        \exp \Bigg\{ 
            \frac{1}{\gamma} \Big[ 
                (S + \epsilon^+) f^+(\epsilon^+) 
                - (S - \epsilon^-) f^-(\epsilon^-) 
                - \eta C^2 
            \Big] 
        \Bigg\}.
    \end{align}

    The objective function is concave if:
    \begin{align}
        (\beta_n^+ - (\alpha_n^+)^2) \, (\beta_n^- - (\alpha_n^-)^2) \geq \delta^2.
    \end{align}
\end{customtheorem}

For a detailed proof, please see Appendix~\ref{appendix: Proof of Main Theorem}.

\section{Choosing the Optimal Radius $\delta$}

Selecting the optimal $\delta$ presents a compelling challenge in our investigation. We employ the principles underlying the Wasserstein robust profile, which offers a nuanced framework for handling uncertainties. This methodology aligns with our goal to enhance the model's resilience against variabilities, thereby illuminating the path toward optimal $\delta$ selection.

We suppose $\boldsymbol{u} = (u^+, u^-)$ is distributed according to the product distribution $\widetilde{\mathbb{Q}}^+ \otimes \widetilde{\mathbb{Q}}^-$ on $\mathbb{R}^2$, and similarly, let $\boldsymbol{\alpha} = (\alpha^+, \alpha^-)$ and $\boldsymbol{\beta} = (\beta^+, \beta^-)$. In the derivation of the optimal policy, it was discovered that the policy is determined solely by the first and second moments of $\xi^{\pm}$. We define the \emph{Wasserstein robust profile} as follows:
\begin{gather}
    \mathcal{R}(\boldsymbol{\alpha}, \boldsymbol{\Sigma}) = \inf \left\{ W_2^2\big(\widetilde{\mathbb{Q}}^+ \otimes \widetilde{\mathbb{Q}}^-, \widehat{\mathbb{Q}}_n^+ \otimes \widehat{\mathbb{Q}}_n^-\big) \mid \mathbb{E}^{\widetilde{\mathbb{Q}}^+ \otimes \widetilde{\mathbb{Q}}^-}[\boldsymbol{u}] = \boldsymbol{\alpha}, \, \mathbb{E}^{\widetilde{\mathbb{Q}}^+ \otimes \widetilde{\mathbb{Q}}^-}[\boldsymbol{u}\boldsymbol{u}^\mathsf{T}] = \boldsymbol{\Sigma} \right\} \\
    \boldsymbol{\Sigma} = \begin{pmatrix}
        \beta^+ & \alpha^+ \alpha^- \\
        \alpha^+ \alpha^- & \beta^-
    \end{pmatrix}
\end{gather}

Additionally, we define the confidence region as follows:
\begin{align} 
\Lambda_\delta =  \biggl\{ \boldsymbol{\pi}( \widetilde{\mathbb{Q}}^\pm) : \widetilde{\mathbb{Q}}^\pm \in \mathcal{U}_{n, \delta}^\pm \biggr\}
\end{align}
where $\boldsymbol{\pi}( \widetilde{\mathbb{Q}}^\pm)$ is the optimal policy for the following optimization problem 
\begin{align}
    &\max_{\boldsymbol{\pi}} \int_{(\mathbb{R}^+)^2} \boldsymbol{\pi}(\epsilon^+, \epsilon^-) \, \mathbb{E}^{\widetilde{\mathbb{Q}}^+ \otimes \widetilde{\mathbb{Q}}^-} \Bigg[ A \Delta \widetilde{N}^+ - B \Delta \widetilde{N}^- \nonumber \\
    & - \eta \left( C + h^{+}(\epsilon^+) \Delta \widetilde{N}^+ - h^{-}(\epsilon^-) \Delta \widetilde{N}^- \right)^2 + (S + \epsilon^+) f^+(\epsilon^+) - (S - \epsilon^-) f^-(\epsilon^-) \Bigg] \, d\epsilon^+ d\epsilon^- \nonumber \\
    & - \gamma \int_{(\mathbb{R}^+)^2} \boldsymbol{\pi}(\epsilon^+, \epsilon^-) \log \boldsymbol{\pi}(\epsilon^+, \epsilon^-) \, d\epsilon^+ d\epsilon^- \label{eqn:main_opt}
\end{align}
As one may notice, this is similar to the robust optimization problem in formula \ref{orig_problem_2}, but with an inner infimum bracket.

It naturally follows that the following relation for the true distribution $\mathbb{Q}_0^\pm$ holds:
\begin{align}
    \boldsymbol{\pi}^* = \boldsymbol{\pi}(\mathbb{Q}_0^\pm) \in \Lambda_\delta \implies \mathcal{R}\left(\boldsymbol{\alpha}^*, \boldsymbol{\Sigma}^*\right) \leq 2\delta^2
\end{align}

where 
\begin{align}
    \boldsymbol{\alpha}^* &= \mathbb{E}^{\mathbb{Q}_0^+ \otimes \mathbb{Q}_0^-}[\boldsymbol{u}], \\
    \boldsymbol{\Sigma}^* &= \mathbb{E}^{\mathbb{Q}_0^+ \otimes \mathbb{Q}_0^-}[\boldsymbol{u} \boldsymbol{u}^\top].
\end{align}

This implication arises because, when $\boldsymbol{\pi}^* \in \Lambda_\delta$, there exists $\widetilde{\mathbb{Q}}^\pm \in \mathcal{U}_\delta(\mathbb{Q}_n^\pm)$ such that $\mathbb{E}^{\widetilde{\mathbb{Q}}^+ \otimes \widetilde{\mathbb{Q}}^-}[\boldsymbol{u}] = \boldsymbol{\alpha}^*$, and $\mathbb{E}^{\widetilde{\mathbb{Q}}^+ \otimes \widetilde{\mathbb{Q}}^-}[\boldsymbol{u}\boldsymbol{u}^\mathsf{T}] = \boldsymbol{\Sigma}^*$. By the definition of $\mathcal{R}(\boldsymbol{\alpha}^*, \boldsymbol{\Sigma}^*)$, we know:
\begin{align}
    \mathcal{R}(\boldsymbol{\alpha}^*, \boldsymbol{\Sigma}^*) \leq   W_2^2(\widetilde{\mathbb{Q}}^+ \otimes \widetilde{\mathbb{Q}}^-, \widehat{\mathbb{Q}}_n^+ \otimes \widehat{\mathbb{Q}}_n^-) = 2\delta^2
\end{align}

To define a confidence region corresponding to the $1 - \chi$ quantile, we select $\delta$ such that:
\begin{align}
    \widehat{\delta}_{1 - \chi} &= \min \left\{ \delta \mid \mathbb{P}(\boldsymbol{\pi}^* \in \Lambda_\delta \mid  \mathbb{Q}_0^\pm) \geq 1 - \chi  \right\} \\
    & \geq \min \left\{ \delta \mid \mathbb{P}(\mathcal{R}(\boldsymbol{\alpha}^*, \boldsymbol{\Sigma}^*) \leq 2\delta^2 \mid  \mathbb{Q}_0^\pm) \geq 1 - \chi  \right\} \\
    & = \min \left\{ \delta \mid \mathbb{P}(\mathcal{R}(\boldsymbol{\alpha}^*, \boldsymbol{\Sigma}^*) > 2\delta^2 \mid  \mathbb{Q}_0^\pm) < \chi  \right\} = \delta_{1 - \chi}
\end{align}

The conditional probability expressions in the equations above describe uncertainty in a statistical sense. Specifically, when we write:
\[
    \mathbb{P}(\boldsymbol{\pi}^* \in \Lambda_\delta \mid \mathbb{Q}_0^\pm)
\]
it means that, given that the true underlying probability distribution is \( \mathbb{Q}_0^\pm \), the probability that \( \boldsymbol{\pi}^* \in \Lambda_\delta \) occurs is at least \( 1 - \chi \). In simpler terms, this quantifies how likely it is that our estimated confidence region \( \Lambda_\delta \) contains the optimal solution \( \boldsymbol{\pi}^* \), assuming that \( \mathbb{Q}_0^\pm \) is indeed the true distribution governing the data.

It is important to note that this is not standard conditional probability notation. While this may not be a conventional way to express conditional probability, we hope the reader finds this notation helpful in conveying the intended meaning.

The theorem presented below delineates the distribution of $\mathcal{R}(\boldsymbol{\alpha}^*, \boldsymbol{\Sigma}^*)$. Based on this theorem, we can determine an optimal radius $\delta$ that enables the formation of a confidence region corresponding to the $1 - \chi$ quantile.

\begin{customtheorem}\label{theorem 2}
    Let \( n \geq 1 \) be given, and let \( \chi \in (0,1) \). Then, provided \( \delta > \frac{c(\chi)}{n} \), it holds that \( \boldsymbol{\pi}^* \in \Lambda_\delta \) with probability at least \( 1 - \chi \) with respect to the draws of the sample data \( \xi_1^\pm, \dots, \xi_n^\pm \). Here, \( c(\chi) \) is the \( \chi \)-quantile of a distribution in $\mathbb{R}$ (described in Appendix).
\end{customtheorem}

For a detailed proof, please refer to Appendix \ref{proof of theorem 2}. The distribution involved is complex and not the primary focus of this work; it is derived from the central limit theorem. The main objective of this theorem is to demonstrate that the convergence rate of the \( 1 - \chi \) confidence radius is \( \mathcal{O}\left( \frac{1}{n} \right) \).


\clearpage
\begin{appendices}

\section{Proof of Theorem \ref{main_thm}}
\label{appendix: Proof of Main Theorem}

The proof of Theorem \ref{main_thm} proceeds in several steps, leveraging the principles of duality and optimization to simplify the problem and establish the optimal policy. Below, we outline the key steps.

\subsection{Step 1: Reformulation Using Duality}
We begin by applying the perfect duality principle, as discussed in \cite{mai2021robust}, to interchange the \(\min\) and \(\max\) operations in the original problem. This reformulation reduces the optimization to:

\begin{align}
    \inf_{\widetilde{\mathbb{Q}}^+ \in \mathcal{U}^+_{n, \delta}, \, \widetilde{\mathbb{Q}}^- \in \mathcal{U}^-_{n, \delta}} \quad & \max_{\boldsymbol{\pi}} \, \int_{(\mathbb{R}^+)^2} \boldsymbol{\pi}(\epsilon^+, \epsilon^-) \, \mathbb{E}^{\widetilde{\mathbb{Q}}^+ \otimes \widetilde{\mathbb{Q}}^-} \Bigg[ A \Delta \widetilde{N}^+ - B \Delta \widetilde{N}^- \nonumber \\
    & \quad - \eta \left( C + h^{+}(\epsilon^+) \Delta \widetilde{N}^+ - h^{-}(\epsilon^-) \Delta \widetilde{N}^- \right)^2 \Bigg] \, d\epsilon^+ d\epsilon^- \nonumber \\
    & \quad - \gamma \int_{(\mathbb{R}^+)^2} \boldsymbol{\pi}(\epsilon^+, \epsilon^-) \log \boldsymbol{\pi}(\epsilon^+, \epsilon^-) \, d\epsilon^+ d\epsilon^-
\end{align}

The optimal policy for the innermost maximization problem can be obtained by applying verification theorem. Consequently, the optimal policy under measure \(\widetilde{\mathbb{Q}}^\pm\) is delineated as follows: 

\begin{align}
    \boldsymbol{\pi}(\epsilon^+, \epsilon^-) 
    &\sim 
    \frac{M(\epsilon^+, \epsilon^-)}{\int_{(\mathbb{R}^+)^2} M(\epsilon^+, \epsilon^-) \, d\epsilon^+ \, d\epsilon^-} \nonumber \\ 
    M(\epsilon^+, \epsilon^-) 
    &= 
    \exp \Bigg\{ 
        \frac{1}{\gamma} \mathbb{E}^{\widetilde{\mathbb{Q}}^+ \otimes \widetilde{\mathbb{Q}}^-} \Big[ 
            (A - 2 \eta C h^+(\epsilon^+)) \Delta \widetilde{N}^+ 
            - (B - 2 \eta C h^-(\epsilon^-)) \Delta \widetilde{N}^- \nonumber \\
    & \quad - \eta \Big( h^+(\epsilon^+) \Delta \widetilde{N}^+ 
            - h^-(\epsilon^-) \Delta \widetilde{N}^- \Big)^2 
        \Big] 
    \Bigg\} L(\epsilon^+, \epsilon^-) \nonumber \\ 
    L(\epsilon^+, \epsilon^-) 
    &= 
    \exp \Bigg\{ 
        \frac{1}{\gamma} \Big[ 
            (S + \epsilon^+) f^+(\epsilon^+) 
            - (S - \epsilon^-) f^-(\epsilon^-) 
            - \eta C^2 
        \Big] 
    \Bigg\}
\end{align}

Upon eliminating terms independent of \((\Delta \widetilde{N}^+, \Delta \widetilde{N}^-)\) and incorporating the optimal policy, the optimization problem we address is formulated as follows:
\begin{align}
     \inf_{\widetilde{\mathbb{Q}}^+ \in \mathcal{U}^+_{n, \delta}, \, \widetilde{\mathbb{Q}}^- \in \mathcal{U}^-_{n, \delta}}  \gamma \int_{(\mathbb{R}^+)^2} M(\epsilon^+, \epsilon^-) \, d\epsilon^+ \, d\epsilon^-
\end{align}
which is equivalent to 
\begin{align}
     \sup_{\widetilde{\mathbb{Q}}^+ \in \mathcal{U}^+_{n, \delta}, \, \widetilde{\mathbb{Q}}^- \in \mathcal{U}^-_{n, \delta}}  -\gamma \int_{(\mathbb{R}^+)^2} M(\epsilon^+, \epsilon^-) \, d\epsilon^+ \, d\epsilon^-
\end{align}
our subsequent objective is to define the constraint domain for the optimization problem. Since \(M(\epsilon^+, \epsilon^-)\) depends on the first and second moments of \(\Delta \widetilde{N}^\pm\), we simplify the notation as follows:
\[
\alpha^\pm = \mathbb{E}^{\widetilde{\mathbb{Q}}^\pm} [\Delta \widetilde{N}^\pm], \quad \beta^\pm = \mathbb{E}^{\widetilde{\mathbb{Q}}^\pm} [(\Delta \widetilde{N}^\pm)^2], \quad \text{where } \widetilde{\mathbb{Q}}^\pm \in \mathcal{U}_{n, \delta}^\pm.
\]
Under the empirical distribution \(\mathbb{\widehat{Q}}_n^\pm\), these quantities are defined as:
\[
\alpha_n^\pm = \mathbb{E}^{\mathbb{\widehat{Q}}_n^\pm} [\Delta \widetilde{N}^\pm], \quad \beta_n^\pm = \mathbb{E}^{\mathbb{\widehat{Q}}_n^\pm} [(\Delta \widetilde{N}^\pm)^2].
\]

Next, we will introduce two propositions before the proof of theorem \ref{main_thm}.

\subsection{Range of \(\alpha^\pm\)}

\begin{proposition}
\label{prop:range_alpha}
The range of \(\alpha^\pm\), where \(\widetilde{\mathbb{Q}}^\pm \in \mathcal{U}_{n, \delta}^\pm\), is given by the interval:
\[
\alpha^\pm \in \left[\alpha_n^\pm - \sqrt{\delta}, \, \alpha_n^\pm + \sqrt{\delta} \right].
\]
\end{proposition}

\begin{proof}
First, compute the upper bound of the expectation
\begin{align}
    \sup_{\widetilde{\mathbb{Q}}^\pm \in \mathcal{U}_\delta(\widehat{\mathbb{Q}}^\pm_n)} \mathbb{E}^{\widetilde{\mathbb{Q}}^\pm}[\xi^\pm] &= \inf_{\lambda} \Big\{ \mathbb E^{\widehat{\mathbb Q}_n^\pm}\big[  \sup_{z} \{  z - \lambda \| z - \xi^\pm   \|^2   \}  \big]  + \lambda \delta \Big\} \nonumber \\
    & = \inf_{\lambda} \Big\{  \mathbb E^{\widehat{\mathbb Q}_n^\pm}\big[  \sup_{\Delta} \{  \Delta + \xi^\pm - \lambda \| \Delta  \|^2   \}  \big]  + \lambda \delta    \Big\} \nonumber \\
    & =  \inf_{\lambda} \Big\{  \mathbb E^{\widehat{\mathbb Q}_n^\pm}\big[ \xi^\pm  \big] + \frac{1}{4\lambda}  + \lambda \delta   \Big\} \nonumber \\
    & =  \mathbb E^{\widehat{\mathbb Q}_n^\pm}\big[ \xi^\pm  \big] + \sqrt{\delta} 
\end{align}
As for the lower bound of the expectation
\begin{align}
    \inf_{\widetilde{\mathbb{Q}}^\pm \in \mathcal{U}_\delta(\widehat{\mathbb{Q}}^\pm_n)} \mathbb{E}^{\widetilde{\mathbb{Q}}^\pm}[\xi^\pm] &= -\sup_{\widetilde{\mathbb{Q}}^\pm \in \mathcal{U}_\delta(\widehat{\mathbb{Q}}^\pm_n)} \mathbb{E}^{\widetilde{\mathbb{Q}}^\pm}[-\xi^\pm] \nonumber \\
    & = - \inf_{\lambda} \Big\{ \mathbb E^{\widehat{\mathbb Q}_n^\pm}\big[  \sup_{z} \{  -z - \lambda \| z - \xi^\pm  \|^2   \}  \big]  + \lambda \delta \Big\} \nonumber \\
    & = -\inf_{\lambda} \Big\{   \mathbb E^{\widehat{\mathbb Q}_n^\pm}\big[  \sup_{\Delta} \{  -\Delta - \xi^\pm - \lambda \| \Delta  \|^2   \}  \big]  + \lambda \delta    \Big\} \nonumber \\ 
    & = - \inf_{\lambda} \Big\{  \mathbb E^{\widehat{\mathbb Q}_n^\pm}\big[ -\xi^\pm  \big] + \frac{1}{4\lambda}  + \lambda \delta   \Big\} \nonumber \\
    & =   \mathbb E^{\widehat{\mathbb Q}_n^\pm}\big[\xi^\pm  \big] - \sqrt{\delta}
\end{align}
\end{proof}

\subsection{Range of $\beta^\pm$}

With this foundation established, we now turn our attention to delineating the range for \(\mathbb{E}^{\widetilde{\mathbb{Q}}^\pm}[(\Delta \widetilde{N}^\pm)^2]\), given that \(\alpha^\pm\) is confined within the interval \(\left[\alpha_n^\pm - \sqrt{\delta}, \alpha_n^\pm + \sqrt{\delta} \right]\). This forms the basis of our second claim.

\begin{proposition}
\label{prop:range_beta}
For a given expectation \(\alpha^\pm \in \left[\alpha_n^\pm - \sqrt{\delta}, \alpha_n^\pm + \sqrt{\delta} \right]\), the second moment satisfies \(\beta^\pm \in [\ell(\alpha^\pm), u(\alpha^\pm)]\), where:
\begin{align}
    u(\alpha^\pm) &= \beta_n^\pm + 2(\alpha^\pm - \alpha_n^\pm) \alpha_n^\pm + \delta + 2 \sqrt{\text{Var}^{\widehat{\mathbb{Q}}^\pm_n}(\Delta \widetilde{N}^\pm)} \sqrt{\delta - (\alpha^\pm - \alpha_n^\pm)^2}, \nonumber \\
    &= \beta_n^\pm + 2(\alpha^\pm - \alpha_n^\pm) \alpha_n^\pm + \delta + 2 \sqrt{\beta_n^\pm - (\alpha_n^\pm)^2} \sqrt{\delta - (\alpha^\pm - \alpha_n^\pm)^2}, \\
    \ell(\alpha^\pm) &= 2(\alpha^\pm - \alpha_n^\pm) \alpha_n^\pm + \delta - 2 \sqrt{\text{Var}^{\widehat{\mathbb{Q}}^\pm_n}(\Delta \widetilde{N}^\pm)} \sqrt{\delta - (\alpha^\pm - \alpha_n^\pm)^2}, \nonumber \\
    &= 2(\alpha^\pm - \alpha_n^\pm) \alpha_n^\pm + \delta - 2 \sqrt{\beta_n^\pm - (\alpha_n^\pm)^2} \sqrt{\delta - (\alpha^\pm - \alpha_n^\pm)^2}.
\end{align}
\end{proposition}

\begin{proof}
 For the upper bound, 
    \begin{align}
        u(\alpha) = \max_{\widetilde{\mathbb{Q}}^\pm \in \mathcal{U}_\delta(\widehat{\mathbb Q}^\pm_n), \  \mathbb{E}^{\widetilde{\mathbb{Q}}^\pm}[\xi^\pm] = \alpha } \mathbb{E}^{\widetilde{\mathbb{Q}}^\pm}[(\xi^\pm)^2]
    \end{align}
    This derivation represents a special case of propositions \(A.2\) and \(A.3\) in \cite{blanchet2022distributionally}, which can be readily deduced. \\ 
    As for the lower-bound
    \begin{align}
        \ell(\alpha) &= - \max_{\widetilde{\mathbb{Q}}^\pm \in \mathcal{U}_\delta(\widehat{\mathbb Q}^\pm_n), \  \mathbb{E}^{\widetilde{\mathbb{Q}}^\pm}[\xi^\pm] = \alpha } \mathbb{E}^{\widetilde{\mathbb{Q}}^\pm}[-(\xi^\pm)^2] \nonumber \\
        & = -\inf_{\lambda_1 \geq 0, \lambda_2} \Big\{ \frac{1}{n} \sum_{i = 1}^{n} \sup_u \big[  -u^2 - \lambda_1 (u - \xi^\pm_i)^2 - \lambda_2 u  \big] + \lambda_1 \delta + \lambda_2 \alpha   \Big\} 
    \end{align}
    Here, 
    \begin{align}
         & \sup_u \big[  -u^2 - \lambda_1 (u - \xi^\pm_i)^2 - \lambda_2 u  \big] \nonumber \\
    = &\sup_u \big[  -(1 + \lambda_1) u^2 + (2\lambda_1 \xi^\pm_i - \lambda_2) u - \lambda_1 (\xi^\pm_i)^2 \big] \nonumber \\
    = &\frac{ (2\lambda_1 \xi^\pm_i - \lambda_2)^2  }{4(1 + \lambda_1)} - \lambda_1 (\xi^\pm_i)^2
    \end{align}
    Then, the lower bound becomes
    \begin{align}
        h(\alpha) =  \inf_{\lambda_1 \geq 0, \lambda_2} \Big\{ \frac{1}{n} \sum_{i = 1}^{n} \bigg[  \frac{ (2\lambda_1 \xi^\pm_i - \lambda_2)^2  }{4(1 + \lambda_1)} - \lambda_1 (\xi^\pm_i)^2 \bigg] + \lambda_1 \delta + \lambda_2 \alpha   \Big\} 
    \end{align}
    Taking the partial derivatives w.r.t $\lambda_2$, we get
    \begin{align}
        \lambda_2 &= -2(1 + \lambda_1) \alpha + 2\lambda_1 \mathbb E^{\widehat{\mathbb Q}^\pm_n} [\xi^\pm]  
    \end{align}
    
    Let $\kappa = 1 + \lambda_1$, and plugging in $\lambda_2$, we get
    \begin{align}
        &\inf_{\kappa \geq 1} \Big\{  \frac{1}{n} \sum_{i = 1}^{n} \Big[  
      \frac{ (\kappa \alpha + (\kappa - 1) \widehat{\xi}^\pm_i - (\kappa - 1) \mathbb E^{\widehat{\mathbb Q}^\pm_n} [\xi^\pm]    )^2 }{ \kappa} - (\kappa - 1) (\widehat{\xi}^\pm_i)^2  \Big] \nonumber \\
      +& (\kappa - 1) \delta - 2\kappa \alpha^2 + 2 (\kappa - 1) \alpha \mathbb E^{\widehat{\mathbb Q}^\pm_n} [\xi^\pm]   \Big\} \nonumber \\
      =& \inf_{\kappa \geq 1} \Big \{  \frac{1}{n} \sum_{i = 1}^{n} \Big[  (\alpha - \mathbb E^{\widehat{\mathbb Q}^\pm_n} [\xi^\pm])(\alpha + 2\widehat{\xi}_i - \mathbb E^{\widehat{\mathbb Q}^\pm_n} [\xi^\pm] ) \kappa + \frac{(\mathbb E^{\widehat{\mathbb Q}^\pm_n} [\xi^\pm]  - \widehat{\xi}_i)^2}{\kappa}  \Big]   + \kappa \delta - 2\alpha^2 \kappa + 2 \alpha \kappa \mathbb E^{\widehat{\mathbb Q}^\pm_n} [\xi^\pm]     \Big \} \nonumber \\
      +& \frac{1}{n} \sum_{i = 1}^{n} \big[ 2 (\alpha + \widehat{\xi}_i - \mathbb E^{\widehat{\mathbb Q}^\pm_n} [\xi^\pm]  )(\mathbb E^{\widehat{\mathbb Q}^\pm_n} [\xi^\pm]  - \widehat{\xi}_i) + (\widehat{\xi}_i)^2 \big] - \delta - 2 \alpha \mathbb E^{\widehat{\mathbb Q}^\pm_n} [\xi^\pm]    \nonumber \\
      =& \inf_{\kappa \geq 1} \Big\{  \big( \delta - (\alpha - \mathbb E^{\widehat{\mathbb Q}^\pm_n} [\xi^\pm])^2   \big) \kappa + \frac{\text{Var}^{\widehat{\mathbb Q}^\pm_n} (\xi)}{\kappa}    \Big\} + 2(\mathbb E^{\widehat{\mathbb Q}^\pm_n} [\xi^\pm])^2  - \delta - 2 \alpha \mathbb E^{\widehat{\mathbb Q}^\pm_n} [\xi^\pm] \nonumber \\
      =& \min \Big\{  \text{Var}^{\widehat{\mathbb Q}^\pm_n}(\xi^\pm)  + (\mathbb E^{\widehat{\mathbb Q}^\pm_n} [\xi^\pm])^2  - \alpha^2, 2(\mathbb E^{\widehat{\mathbb Q}^\pm_n} [\xi^\pm])^2  - \delta - 2 \alpha \mathbb E^{\widehat{\mathbb Q}^\pm_n} [\xi^\pm] \nonumber \\
      +& 2 \sqrt{\text{Var}^{\widehat{\mathbb Q}^\pm_n} (\xi^\pm)} \sqrt{ \delta - (\alpha - \mathbb E^{\widehat{\mathbb Q}^\pm_n} [\xi^\pm])^2 } \   \Big\}
    \end{align}
    As we notice that, in the above $\min$ bracket, 
    \begin{align}
         &\delta + 2 \alpha \mathbb E^{\widehat{\mathbb Q}^\pm_n} [\xi^\pm] - 2 \sqrt{\text{Var}^{\widehat{\mathbb Q}^\pm_n} (\xi)} \sqrt{ \delta - (\alpha - \mathbb E^{\widehat{\mathbb Q}^\pm_n} [\xi^\pm])^2  } - 2(\mathbb E^{\widehat{\mathbb Q}^\pm_n} [\xi^\pm])^2 - (\alpha^2 -   \text{Var}^{\widehat{\mathbb Q}^\pm_n}(\xi^\pm)  - (\mathbb E^{\widehat{\mathbb Q}^\pm_n} [\xi^\pm])^2 ) \nonumber \\
          = & \ \delta -  2 \sqrt{\text{Var}^{\widehat{\mathbb Q}^\pm_n} (\xi)} \sqrt{ \delta - (\alpha - \mathbb E^{\widehat{\mathbb Q}^\pm_n} [\xi^\pm])^2  } + \text{Var}^{\widehat{\mathbb Q}^\pm_n}(\xi) - (\alpha - \mathbb E^{\widehat{\mathbb Q}^\pm_n} [\xi^\pm])^2 \nonumber \\
          = & \ \Big( \sqrt{\delta - (\alpha - \mathbb E^{\widehat{\mathbb Q}^\pm_n} [\xi^\pm])^2 } - \sqrt{\text{Var}^{\widehat{\mathbb Q}^\pm_n} (\xi)}  \Big)^2 \geq 0
    \end{align}
    Therefore, the lower bound is
    \begin{align}
        \ell(\alpha) = 2(\alpha - \mathbb E^{\widehat{\mathbb Q}^\pm_n}[\xi^\pm]) \mathbb E^{\widehat{\mathbb Q}^\pm_n}[\xi^\pm] + \delta  - 2 \sqrt{\text{Var}^{\widehat{\mathbb Q}^\pm_n} (\xi)} \sqrt{ \delta - (\alpha - \mathbb E^{\widehat{\mathbb Q}^\pm_n} [\xi^\pm])^2  } 
    \end{align}
\end{proof}

\subsection{Proof of Theorem \ref{main_thm}}
With the proposition \ref{prop:range_alpha}, and proposition \ref{prop:range_beta}, we are now ready for the proof of the theorem \ref{main_thm}.

\begin{proof}
After plugging the $\beta^\pm$, the expression inside the exponential function becomes 
    \begin{align}
        &(A - 2 \widetilde{\eta} C h^+) \alpha^+ - (B - 2 \widetilde{\eta} C h^-) \alpha^- + 2\widetilde{\eta} h^+ h^- \alpha^+ \alpha^- - \widetilde{\eta} (h^+)^2 \beta^+ - \widetilde{\eta} (h^-)^2 \beta^- \nonumber \\
        = & (A - 2 \widetilde{\eta} C h^+) \alpha^+ - (B - 2 \widetilde{\eta} C h^-) \alpha^- - \widetilde{\eta}(h^+\alpha^+ - h^- \alpha^-)^2 \nonumber \\
        - & \widetilde{\eta} (h^+)^2 \Big( \sqrt{\text{Var}^{\widehat{\mathbb Q}^+_n} (\xi^+)} +  \sqrt{\delta - (\alpha^+ - \mathbb E^{\widehat{\mathbb Q}^+_n}[\xi^+])^2}  \Big)^2 - \widetilde{\eta} (h^-)^2 \Big( \sqrt{\text{Var}^{\widehat{\mathbb Q}^-_n} (\xi^-)} +  \sqrt{\delta - (\alpha^- - \mathbb E^{\widehat{\mathbb Q}^-_n}[\xi^\pm])^2}  \Big)^2 \nonumber 
    \end{align}
    The following function $(h^+\alpha^+ - h^- \alpha^-)^2$ has Hessian matrix 
    \begin{align}
        \begin{pmatrix}
             2(h^+)^2 & -2h^+h^- \\
             -2h^+h^- & 2(h^-)^2
        \end{pmatrix}
    \end{align}
    the Hessian matrix has determinant $0$, then by Sylvester's criterion, it is a positive semi-definite matrix, which makes $(h^+\alpha^+ - h^- \alpha^-)^2$ a convex function. 
    now, let's consider the last two terms, first define function $\phi(\alpha^\pm)$
    \begin{align}
        \phi(\alpha^\pm) = \Big( \sqrt{\text{Var}^{\widehat{\mathbb Q}^\pm_n} (\xi^\pm)} +  \sqrt{\delta - (\alpha^\pm - \mathbb E^{\widehat{\mathbb Q}^\pm_n}[\xi^\pm])^2}  \Big)^2
    \end{align}
    the first derivative of $\phi(\alpha^\pm)$ is 
    \begin{align}
        \frac{\partial \phi(\alpha^\pm)}{\partial \alpha^\pm} = 2 \Big( \sqrt{\text{Var}^{\widehat{\mathbb Q}^\pm_n} (\xi^\pm)} +  \sqrt{\delta - (\alpha^\pm - \mathbb E^{\widehat{\mathbb Q}^\pm_n}[\xi^\pm])^2}  \Big) \frac{- (\alpha^\pm - \mathbb E^{\widehat{\mathbb Q}^\pm_n}[\xi^\pm]) }{\sqrt{\delta - (\alpha^\pm - \mathbb E^{\widehat{\mathbb Q}^\pm_n}[\xi^\pm])^2} }
    \end{align}
    then the second derivative is 
    \begin{align}
        &\frac{\partial}{\partial \alpha^\pm} \  2 \Big( \sqrt{\text{Var}^{\widehat{\mathbb Q}^\pm_n} (\xi^\pm)} +  \sqrt{\delta - (\alpha^\pm - \mathbb E^{\widehat{\mathbb Q}^\pm_n}[\xi^\pm])^2}  \Big) \frac{- (\alpha^\pm - \mathbb E^{\widehat{\mathbb Q}^\pm_n}[\xi^\pm]) }{\sqrt{\delta - (\alpha^\pm - \mathbb E^{\widehat{\mathbb Q}^\pm_n}[\xi^\pm])^2} } \nonumber \\
        =&  2 \Big( \sqrt{\text{Var}^{\widehat{\mathbb Q}^\pm_n} (\xi)} +  \sqrt{\delta - (\alpha^\pm - \mathbb E^{\widehat{\mathbb Q}^\pm_n}[\xi^\pm])^2}  \Big) \frac{ - \sqrt{\delta - (\alpha^\pm - \mathbb E^{\widehat{\mathbb Q}^\pm_n}[\xi^\pm])^2} -   \frac{(\alpha^\pm - \mathbb E^{\widehat{\mathbb Q}^\pm_n}[\xi^\pm])^2 }{\sqrt{\delta - (\alpha^\pm - \mathbb E^{\widehat{\mathbb Q}^\pm_n}[\xi^\pm])^2} } }{\delta - (\alpha^\pm - \mathbb E^{\widehat{\mathbb Q}^\pm_n}[\xi^\pm])^2} \nonumber \\
        +& \frac{2 (\alpha^\pm - \mathbb E^{\widehat{\mathbb Q}^\pm_n}[\xi^\pm])^2 }{\delta - (\alpha^\pm - \mathbb E^{\widehat{\mathbb Q}^\pm_n}[\xi^\pm])^2} \nonumber \\
    =& \frac{2 (\alpha^\pm - \mathbb E^{\widehat{\mathbb Q}^\pm_n}[\xi^\pm])^2 }{\delta - (\alpha^\pm - \mathbb E^{\widehat{\mathbb Q}^\pm_n}[\xi^\pm])^2} -  \frac{ 2\delta \Big( \sqrt{\text{Var}^{\widehat{\mathbb Q}^\pm_n} (\xi)} +  \sqrt{\delta - (\alpha^\pm - \mathbb E^{\widehat{\mathbb Q}^\pm_n}[\xi^\pm])^2}  \Big) }{\Big(\delta - (\alpha^\pm - \mathbb E^{\widehat{\mathbb Q}^\pm_n}[\xi^\pm])^2\Big)^{\frac{3}{2}}} \nonumber \\
    =& \frac{2\Big( \big[ ( \alpha^\pm - \mathbb E^{\widehat{\mathbb Q}^\pm_n}[\xi^\pm]  )^2 - \delta  \big] \sqrt{\delta - (\alpha^\pm - \mathbb E^{\widehat{\mathbb Q}^\pm_n}[\xi^\pm])^2} - \delta  \sqrt{\text{Var}^{\widehat{\mathbb Q}^\pm_n} (\xi^\pm)}  \Big)}{\Big(\delta - (\alpha^\pm - \mathbb E^{\widehat{\mathbb Q}^\pm_n}[\xi^\pm])^2\Big)^{\frac{3}{2}}} \nonumber \\
    =& -2 - \frac{2 \delta \sqrt{\text{Var}^{\widehat{\mathbb Q}^\pm_n} (\xi^\pm)}  \Big)}{\Big(\delta - (\alpha^\pm - \mathbb E^{\widehat{\mathbb Q}^\pm_n}[\xi^\pm])^2\Big)^{\frac{3}{2}}}
    \end{align}
    notice that $ \mathbb{E}^{\widehat{\mathbb{Q}}^\pm_n}[\xi^\pm] -\sqrt{\delta} \leq \alpha^\pm \leq \mathbb{E}^{\widehat{\mathbb{Q}}^\pm_n}[\xi^\pm] + \sqrt{\delta}$, we know that the second derivative is negative.

    Therefore, consider the following function
    \begin{align}
        \widetilde{\eta}(h^+\alpha^+ - h^- \alpha^-)^2 + \widetilde{\eta}(h^+)^2 \phi(\alpha^+) + \widetilde{\eta}(h^-)^2 \phi(\alpha^-)
    \end{align}
    the Hessian matrix of the above function is 
    \begin{align}
        &\begin{pmatrix}
            2\widetilde{\eta}(h^+)^2 - 2 \widetilde{\eta} (h^+)^2 -  \frac{2 \delta \widetilde{\eta} (h^+)^2 \sqrt{\text{Var}^{\widehat{\mathbb Q}^+_n} (\xi^+)}  \Big)}{\Big(\delta - (\alpha^+ - \mathbb E^{\widehat{\mathbb Q}^+_n}[\xi^+])^2\Big)^{\frac{3}{2}}}& -2\widetilde{\eta}h^+h^- \\  -2\widetilde{\eta} h^+h^- 
             & 2\widetilde{\eta} (h^-)^2 - 2 \widetilde{\eta}(h^-)^2 -  \frac{2  \delta \widetilde{\eta} (h^-)^2 \sqrt{\text{Var}^{\widehat{\mathbb Q}^-_n} (\xi^-)}  \Big)}{\Big(\delta - (\alpha^- - \mathbb E^{\widehat{\mathbb Q}^-_n}[\xi^-])^2\Big)^{\frac{3}{2}}}
        \end{pmatrix} \nonumber \\
        =& \begin{pmatrix}
            -  \frac{2 \delta \widetilde{\eta} (h^+)^2 \sqrt{\text{Var}^{\widehat{\mathbb Q}^+_n} (\xi^+)}  \Big)}{\Big(\delta - (\alpha^+ - \mathbb E^{\widehat{\mathbb Q}^+_n}[\xi^+])^2\Big)^{\frac{3}{2}}} &  -2\widetilde{\eta}h^+h^- \\
             -2\widetilde{\eta}h^+h^- &  -\frac{2  \delta \widetilde{\eta} (h^-)^2 \sqrt{\text{Var}^{\widehat{\mathbb Q}^-_n} (\xi^-)}  \Big)}{\Big(\delta - (\alpha^- - \mathbb E^{\widehat{\mathbb Q}^-_n}[\xi^-])^2\Big)^{\frac{3}{2}}}
        \end{pmatrix}
    \end{align}

    Thus, the original equation has Hessian matrix 
    \begin{align}
        \begin{pmatrix}
             \frac{2 \delta \widetilde{\eta} (h^+)^2 \sqrt{\text{Var}^{\widehat{\mathbb Q}^+_n} (\xi^+)}  \Big)}{\Big(\delta - (\alpha^+ - \mathbb E^{\widehat{\mathbb Q}^+_n}[\xi^+])^2\Big)^{\frac{3}{2}}} &  2\widetilde{\eta}h^+h^- \\
             2\widetilde{\eta}h^+h^- &  \frac{2  \delta \widetilde{\eta} (h^-)^2 \sqrt{\text{Var}^{\widehat{\mathbb Q}^-_n} (\xi^-)}  \Big)}{\Big(\delta - (\alpha^- - \mathbb E^{\widehat{\mathbb Q}^-_n}[\xi^-])^2\Big)^{\frac{3}{2}}}
        \end{pmatrix}
    \end{align}

    The determinant of the above matrix is
    \begin{align}
        4\widetilde{\eta}^2 (h^+)^2(h^-)^2 \Bigg( \frac{\delta^2 \sqrt{\text{Var}^{\widehat{\mathbb Q}^+_n} (\xi^+)} \sqrt{\text{Var}^{\widehat{\mathbb Q}^-_n} (\xi^-)} }{\Big(\delta - (\alpha^+ - \mathbb E^{\widehat{\mathbb Q}^+_n}[\xi^+])^2\Big)^{\frac{3}{2}} \Big(\delta - (\alpha^- - \mathbb E^{\widehat{\mathbb Q}^-_n}[\xi^-])^2\Big)^{\frac{3}{2}}  } - 1   \Bigg)
    \end{align}

    When $\alpha^\pm = \mathbb E^{\widehat{\mathbb Q}^\pm_n}[\xi^\pm]$, the above determinant becomes the lowest, which is 
    \begin{align}
           4\widetilde{\eta}^2 (h^+)^2(h^-)^2 \Bigg( \frac{ \sqrt{\text{Var}^{\widehat{\mathbb Q}^+_n} (\xi^+)} \sqrt{\text{Var}^{\widehat{\mathbb Q}^-_n} (\xi^-)} }{\delta} - 1   \Bigg)
    \end{align}

    Thus, according to the assumption, the function inside the exponential bracket is convex.  More specifically, the following function of $\alpha^\pm$ is convex
    \begin{align}
        (A - 2 \widetilde{\eta} C h^+) \alpha^+ - (B - 2 \widetilde{\eta} C h^-) \alpha^- + 2\widetilde{\eta} h^+ h^- \alpha^+ \alpha^- - \widetilde{\eta} (h^+)^2 \beta^+ - \widetilde{\eta} (h^-)^2 \beta^-
    \end{align}
    is a convex function w.r.t $\alpha^\pm$

    To show the objective function is concave when $\text{Var}^{\widehat{\mathbb Q}^+_n} (\xi^+)\text{Var}^{\widehat{\mathbb Q}^-_n} (\xi^-) \geq \delta^2$,
    let's consider the following, assume $p(\boldsymbol{x})$ is a convex function, then $e^{p(\boldsymbol{x})}$ is also a convex function 
    Since $p(\boldsymbol{x})$ is a convex function, then 
    \begin{align}
        p(\lambda \boldsymbol{x} + (1 - \lambda) \boldsymbol{y}) \leq \lambda p(\boldsymbol{x}) + (1 - \lambda) p(\boldsymbol{y})
    \end{align}
    because the exponential function is increasing and convex, then 
    \begin{align}
        e^{p(\lambda \boldsymbol{x} + (1 - \lambda) \boldsymbol{y} )} \leq e^{\lambda p(\boldsymbol{x}) + (1 - \lambda) p(\boldsymbol{y})} \leq  \lambda e^{p(\boldsymbol{x})} + (1 - \lambda) e^{p(\boldsymbol{y})}
    \end{align}
    Since integral doesn't change the convexity or concavity, then the objective function is concave 
\end{proof}

\section{Proof of Theorem \ref{theorem 2}} \label{proof of theorem 2}

\begin{proof}
\begin{align}
 \mathcal{R}(\boldsymbol{\alpha}, \boldsymbol{\Sigma}) = \inf \Big\{ W_2^2\big(\widetilde{\mathbb{Q}}^+ \otimes \widetilde{\mathbb{Q}}^-, \widehat{\mathbb{Q}}_n^+ \otimes \widehat{\mathbb{Q}}_n^-\big) \, \Big| \, \mathbb{E}^{\widetilde{\mathbb{Q}}^+ \otimes \widetilde{\mathbb{Q}}^-}[\boldsymbol{u}] = \boldsymbol{\alpha}, \, \mathbb{E}^{\widetilde{\mathbb{Q}}^+ \otimes \widetilde{\mathbb{Q}}^-}[\boldsymbol{u}\boldsymbol{u}^\mathsf{T}] = \boldsymbol{\Sigma} \Big\}
\end{align}

Thus, following the derivation in \cite{blanchet2022distributionally}, we know the duality of $\mathcal{R}(\boldsymbol{\alpha}, \boldsymbol{\Sigma})$ is 
\begin{align}
    \mathcal{R}(\boldsymbol{\alpha}, \boldsymbol{\Sigma}) =& \sup_{\boldsymbol \Lambda \in \mathbb R^{2 \times 2}, \  \boldsymbol \lambda \in \mathbb R^2} \bigg\{  - \mathbb E^{\widehat{\mathbb{Q}}_n^+ \otimes \widehat{\mathbb{Q}}_n^-} \Big[ \sup_{\boldsymbol{\phi} \in \mathbb R^2} \Big\{ \textbf{Tr} \big(\Lambda [\boldsymbol{\phi}\boldsymbol{\phi}^T - \Sigma ] \big) + \boldsymbol{\lambda}^T (\boldsymbol{\phi} - \boldsymbol{\alpha}) - \| \boldsymbol{\phi} - \boldsymbol{u} \|_2^2  \Big\}  \Big]  \bigg\}
\end{align}

Since 
\begin{align}
    &\sup_{\boldsymbol{\phi} \in \mathbb R^2} \Big\{ \textbf{Tr} \big(\Lambda [\boldsymbol{\phi}\boldsymbol{\phi}^T - \Sigma ] \big) + \boldsymbol{\lambda}^T (\boldsymbol{\phi} - \boldsymbol{\alpha}) - \| \boldsymbol{\phi} - \boldsymbol{u} \|_2^2  \Big\} \nonumber \\
    =& \sup_{\boldsymbol{\Delta} \in \mathbb R^2} \Big\{ \textbf{Tr} \big(\Lambda [(\boldsymbol{\Delta} + \boldsymbol{u}) (\boldsymbol{\Delta} + \boldsymbol{u})^T - \Sigma ] \big) + \boldsymbol{\lambda}^T (\boldsymbol{\Delta} + \boldsymbol{u} - \boldsymbol{\alpha})   - \| \boldsymbol{\Delta} \|_2^2  \Big\} \nonumber \\
    =&  \sup_{\boldsymbol{\Delta} \in \mathbb R^2} \Big\{ \textbf{Tr} \big(\Lambda [(\boldsymbol{\Delta} + \boldsymbol{u}) (\boldsymbol{\Delta} + \boldsymbol{u})^T - \boldsymbol{u}\boldsymbol{u}^T ] \big) + \boldsymbol{\lambda}^T \boldsymbol{\Delta}   - \| \boldsymbol{\Delta} \|_2^2  \Big\}  \nonumber \\
    +& \textbf{Tr}(\boldsymbol{\Lambda}\big[ \boldsymbol{u}\boldsymbol{u}^T - \Sigma \big]) + \boldsymbol{\lambda}^T (\boldsymbol{u} - \boldsymbol{\alpha}) \nonumber \\
\end{align}

For the last term, it can be written as
\begin{align}
    \textbf{Tr} \big(\boldsymbol \Lambda [(\boldsymbol{\Delta} + \boldsymbol{u}) (\boldsymbol{\Delta} + \boldsymbol{u})^T - \boldsymbol{u}\boldsymbol{u}^T ] \big) =& \int_0^1 \Big( 2\textbf{Tr} \big( \boldsymbol{\Lambda} \boldsymbol{u} \boldsymbol{\Delta}^T \big) + 2t \boldsymbol \Delta^T \boldsymbol{\Lambda} \boldsymbol{\Delta} \Big) \ dt \nonumber \\
    =& \Big( 2\textbf{Tr} \big( \boldsymbol{\Lambda} \boldsymbol{u} \boldsymbol{\Delta}^T \big) +  \boldsymbol \Delta^T \boldsymbol{\Lambda} \boldsymbol{\Delta} \Big)
\end{align}

By plugging $\boldsymbol{\alpha}^*$,$\boldsymbol{\Sigma}^*$ to the above formula, there is
    \begin{align}
    \mathcal{R}(\boldsymbol{\alpha}^*, \boldsymbol{\Sigma}^*) =& \sup_{\boldsymbol{\lambda} \in \mathbb R^2} \Bigg\{ -\mathbb E^{\widehat{\mathbb{Q}}_n^+ \otimes \widehat{\mathbb{Q}}_n^-} \big[ \boldsymbol{\lambda^T (\boldsymbol{u} - \boldsymbol{\alpha}^*)} \big]  \nonumber \\
    +& \sup_{\boldsymbol{\Lambda} \in \mathbb R^{2 \times 2}} \bigg\{  -\mathbb E^{\widehat{\mathbb{Q}}_n^+ \otimes \widehat{\mathbb{Q}}_n^-} \Big[  \sup_{\boldsymbol{\Delta} \in \mathbb R^2} \Big\{ 2\textbf{Tr} \big( \boldsymbol{\Lambda} \boldsymbol{u} \boldsymbol{\Delta}^T \big) + \boldsymbol \Delta^T \boldsymbol{\Lambda} \boldsymbol{\Delta} + \boldsymbol{\lambda}^T \boldsymbol{\Delta}   - \| \boldsymbol{\Delta} \|_2^2  \Big\} \Big] - \textbf{Tr}\big( \boldsymbol \Lambda (\boldsymbol{\Sigma}_n - \boldsymbol{\Sigma}^*) \big) \bigg\} \Bigg\} \nonumber  \\ 
    =&  \sup_{\boldsymbol{\lambda} \in \mathbb R^2} \Big\{ -\mathbb E^{\widehat{\mathbb{Q}}_n^+ \otimes \widehat{\mathbb{Q}}_n^-} \big[ \boldsymbol{\lambda^T (\boldsymbol{u} - \boldsymbol{\alpha}^*)} \big] \Big\} + \sup_{\boldsymbol{\Lambda} \in \mathbb R^{2 \times 2}} \bigg\{  -\mathbb E^{\widehat{\mathbb{Q}}_n^+ \otimes \widehat{\mathbb{Q}}_n^-} \Big[ n \| \boldsymbol{\Lambda} \boldsymbol{u} + \boldsymbol{\lambda}  \|_2^2  \Big] - \textbf{Tr}\big( \boldsymbol \Lambda (\boldsymbol{\Sigma}_n - \boldsymbol{\Sigma}^*) \big) \bigg\}  \nonumber \\
    =&  \sup_{\boldsymbol{\lambda} \in \mathbb R^2} \Big\{ -\mathbb E^{\widehat{\mathbb{Q}}_n^+ \otimes \widehat{\mathbb{Q}}_n^-} \big[ \boldsymbol{\lambda^T (\boldsymbol{u} - \boldsymbol{\alpha}^*)} \big] \Big\} - \inf_{\boldsymbol{\Lambda} \in \mathbb R^{2 \times 2}} \bigg\{  \mathbb E^{\widehat{\mathbb{Q}}_n^+ \otimes \widehat{\mathbb{Q}}_n^-} \Big[ n \| \boldsymbol{\Lambda} \boldsymbol{u} + \boldsymbol{\lambda}  \|_2^2  \Big] + \textbf{Tr}\big( \boldsymbol \Lambda (\boldsymbol{\Sigma}_n - \boldsymbol{\Sigma}^*) \big) \bigg\}  \nonumber 
\end{align}

Now, we take derivative with respect to $\boldsymbol{\Lambda}$
\begin{align}
     &\nabla_{\boldsymbol{\Lambda}} \mathbb E^{\widehat{\mathbb{Q}}_n^+ \otimes \widehat{\mathbb{Q}}_n^-} \Big[ n \| \boldsymbol{\Lambda} \boldsymbol{u} + \boldsymbol{\lambda}  \|_2^2  \Big] + \textbf{Tr}\big( \boldsymbol \Lambda (\boldsymbol{\Sigma}_n - \boldsymbol{\Sigma}^*) \big)  \nonumber \\ 
     =& 2 n E^{\widehat{\mathbb{Q}}_n^+ \otimes \widehat{\mathbb{Q}}_n^-} \Big[  ( \boldsymbol{\Lambda} \boldsymbol{u} + \boldsymbol{\lambda} ) \boldsymbol{u}^T \Big] + (\boldsymbol{\Sigma}_n - \boldsymbol{\Sigma}^*) \nonumber \\
     =& 2n\boldsymbol{\Lambda} \boldsymbol{\Sigma}_n + 2n\boldsymbol{\lambda} \boldsymbol{\alpha}^T_n + (\boldsymbol{\Sigma}_n - \boldsymbol{\Sigma}^*) = \boldsymbol{0}
\end{align}

which results to 
\begin{align}
    \boldsymbol{\Lambda}^* =& -\frac{1}{2n}  \big(  2n\boldsymbol{\lambda} \boldsymbol{\alpha}^T_n + (\boldsymbol{\Sigma}_n - \boldsymbol{\Sigma}^*) \big) \boldsymbol{\Sigma}_n^{-1} \nonumber \\
    =&  - \big[\boldsymbol{\lambda} \boldsymbol{\alpha}_n^T  + \frac{1}{2n} (\boldsymbol{\Sigma}_n - \boldsymbol{\Sigma}^*) \big]\boldsymbol{\Sigma}_n^{-1} \nonumber 
\end{align}

Plug the optimal $\boldsymbol{\Lambda}$, we obtain
\begin{align}
     &\mathbb E^{\widehat{\mathbb{Q}}_n^+ \otimes \widehat{\mathbb{Q}}_n^-} \Big[ \| \boldsymbol{\Lambda}^* \boldsymbol{u} + \boldsymbol{\lambda}  \|_2^2  \Big] \nonumber \\
     =& \boldsymbol{\lambda}^T \boldsymbol{\lambda} + 2\boldsymbol{\lambda}^T \boldsymbol{\Lambda}^* \boldsymbol{\alpha}_n  + \mathbb E^{\widehat{\mathbb{Q}}_n^+ \otimes \widehat{\mathbb{Q}}_n^-} \Big[ \boldsymbol{u}^T (\boldsymbol{\Lambda}^{*})^T \boldsymbol{\Lambda}^* \boldsymbol{u}  \Big] \nonumber
\end{align}

Since
\begin{align}
    (\boldsymbol{\Lambda}^{*})^T \boldsymbol{\Lambda}^* &= \boldsymbol{\Sigma}_n^{-1} \big(  \boldsymbol{\alpha}_n \boldsymbol{\lambda}^T  + \frac{1}{2n} (\boldsymbol{\Sigma}_n - \boldsymbol{\Sigma}^*) \big)  \big(  \boldsymbol{\lambda} \boldsymbol{\alpha}_n^T   + \frac{1}{2n} (\boldsymbol{\Sigma}_n - \boldsymbol{\Sigma}^*) \big) \boldsymbol{\Sigma}_n^{-1} \nonumber \\
    &= \boldsymbol{\lambda}^T \boldsymbol{\lambda} \big( \boldsymbol{\Sigma}_n^{-1} \boldsymbol{\alpha}_n \boldsymbol{\alpha}_n^T \boldsymbol{\Sigma}_n^{-1} \big)  + \frac{1}{n} \boldsymbol{\Sigma}_n^{-1} \boldsymbol{\alpha}_n \boldsymbol{\lambda}^T (\boldsymbol{\Sigma}_n - \boldsymbol{\Sigma}^*)\boldsymbol{\Sigma}_n^{-1} + \frac{1}{4n^2} \boldsymbol{\Sigma}_n^{-1}(\boldsymbol{\Sigma}_n - \boldsymbol{\Sigma}^*)\boldsymbol{\Sigma}_n^{-1}
\end{align}

Thus, 
\begin{align}
     &\boldsymbol{u}^T(\boldsymbol{\Lambda}^{*})^T \boldsymbol{\Lambda}^* \boldsymbol{u} \nonumber \\
     =& \boldsymbol{\lambda}^T \boldsymbol{\lambda} \big( \boldsymbol{u}^T \boldsymbol{\Sigma}_n^{-1} \boldsymbol{\alpha}_n \boldsymbol{\alpha}_n^T \boldsymbol{\Sigma}_n^{-1} \boldsymbol{u} \big) + \frac{1}{n} \boldsymbol{u}^T\boldsymbol{\Sigma}_n^{-1} \boldsymbol{\alpha}_n \boldsymbol{\lambda}^T (\boldsymbol{\Sigma}_n - \boldsymbol{\Sigma}^*)\boldsymbol{\Sigma}_n^{-1} \boldsymbol{u} + \frac{1}{4n^2} \boldsymbol{u}^T\boldsymbol{\Sigma}_n^{-1}(\boldsymbol{\Sigma}_n - \boldsymbol{\Sigma}^*)\boldsymbol{\Sigma}_n^{-1}\boldsymbol{u}   \nonumber \\
     =& \boldsymbol{\lambda}^T \boldsymbol{\lambda} \big(\boldsymbol{\alpha}_n^T \boldsymbol{\Sigma}_n^{-1} \boldsymbol{u} \boldsymbol{u}^T \boldsymbol{\Sigma}_n^{-1} \boldsymbol{\alpha}_n  \big)  + \frac{1}{n} \boldsymbol{\lambda}^T (\boldsymbol{\Sigma}_n - \boldsymbol{\Sigma}^*)\boldsymbol{\Sigma}_n^{-1} \boldsymbol{u} \boldsymbol{u}^T\boldsymbol{\Sigma}_n^{-1} \boldsymbol{\alpha}_n + \frac{1}{4n^2} \boldsymbol{u}^T\boldsymbol{\Sigma}_n^{-1}(\boldsymbol{\Sigma}_n - \boldsymbol{\Sigma}^*)\boldsymbol{\Sigma}_n^{-1}\boldsymbol{u} 
\end{align}

Then, taking the expectation, we have 
\begin{align}
    &\mathbb E^{\widehat{\mathbb{Q}}_n^+ \otimes \widehat{\mathbb{Q}}_n^-} \Big[ \boldsymbol{u}^T (\boldsymbol{\Lambda}^{*})^T \boldsymbol{\Lambda}^* \boldsymbol{u}  \Big] \nonumber \\
    =&  \boldsymbol{\lambda}^T \boldsymbol{\lambda} \big(\boldsymbol{\alpha}_n^T \boldsymbol{\Sigma}_n^{-1} \mathbb E^{\widehat{\mathbb{Q}}_n^+ \otimes \widehat{\mathbb{Q}}_n^-} \big[\boldsymbol{u} \boldsymbol{u}^T\big] \boldsymbol{\Sigma}_n^{-1} \boldsymbol{\alpha}_n  \big)  + \frac{1}{n} \boldsymbol{\lambda}^T (\boldsymbol{\Sigma}_n - \boldsymbol{\Sigma}^*)\boldsymbol{\Sigma}_n^{-1} \mathbb E^{\widehat{\mathbb{Q}}_n^+ \otimes \widehat{\mathbb{Q}}_n^-}\big[\boldsymbol{u} \boldsymbol{u}^T\big]\boldsymbol{\Sigma}_n^{-1} \boldsymbol{\alpha}_n \nonumber \\
    +& \frac{1}{4n^2} \mathbb E^{\widehat{\mathbb{Q}}_n^+ \otimes \widehat{\mathbb{Q}}_n^-}\big[ \boldsymbol{u}^T\boldsymbol{\Sigma}_n^{-1}(\boldsymbol{\Sigma}_n - \boldsymbol{\Sigma}^*)\boldsymbol{\Sigma}_n^{-1}\boldsymbol{u} \big] \nonumber \\
    =&  \boldsymbol{\lambda}^T \boldsymbol{\lambda} \big(\boldsymbol{\alpha}_n^T \boldsymbol{\Sigma}_n^{-1} \boldsymbol{\Sigma}_n \boldsymbol{\Sigma}_n^{-1} \boldsymbol{\alpha}_n  \big)  + \frac{1}{n} \boldsymbol{\lambda}^T (\boldsymbol{\Sigma}_n - \boldsymbol{\Sigma}^*)\boldsymbol{\Sigma}_n^{-1} \boldsymbol{\Sigma}_n \boldsymbol{\Sigma}_n^{-1} \boldsymbol{\alpha}_n \nonumber \\
    +& \frac{1}{4n^2} \mathbb E^{\widehat{\mathbb{Q}}_n^+ \otimes \widehat{\mathbb{Q}}_n^-}\big[ \boldsymbol{u}^T\boldsymbol{\Sigma}_n^{-1}(\boldsymbol{\Sigma}_n - \boldsymbol{\Sigma}^*)\boldsymbol{\Sigma}_n^{-1}\boldsymbol{u} \big] \nonumber \\
    =& \boldsymbol{\lambda}^T \boldsymbol{\lambda} \big(\boldsymbol{\alpha}_n^T \boldsymbol{\Sigma}_n^{-1} \boldsymbol{\alpha}_n  \big)  + \frac{1}{n} \boldsymbol{\lambda}^T (\boldsymbol{\Sigma}_n - \boldsymbol{\Sigma}^*)\boldsymbol{\Sigma}_n^{-1} \boldsymbol{\alpha}_n + \frac{1}{4n^2} \mathbb E^{\widehat{\mathbb{Q}}_n^+ \otimes \widehat{\mathbb{Q}}_n^-}\big[ \boldsymbol{u}^T\boldsymbol{\Sigma}_n^{-1}(\boldsymbol{\Sigma}_n - \boldsymbol{\Sigma}^*)\boldsymbol{\Sigma}_n^{-1}\boldsymbol{u} \big] \nonumber 
\end{align}

Next, let's consider the second term 
\begin{align}
    &2\boldsymbol{\lambda}^T \boldsymbol{\Lambda}^* \boldsymbol{\alpha}_n \nonumber \\
    =& -2\boldsymbol{\lambda}^T \big[\boldsymbol{\lambda} \boldsymbol{\alpha}_n^T  + \frac{1}{2n} (\boldsymbol{\Sigma}_n - \boldsymbol{\Sigma}^*) \big]\boldsymbol{\Sigma}_n^{-1} \boldsymbol{\alpha}_n  \nonumber \\
    =& - 2 \boldsymbol{\lambda}^T \boldsymbol{\lambda}  \boldsymbol{\alpha}_n^T \boldsymbol{\Sigma}_n^{-1} \boldsymbol{\alpha}_n  - \frac{1}{n} \boldsymbol{\lambda}^T (\boldsymbol{\Sigma}_n - \boldsymbol{\Sigma}^*)\boldsymbol{\Sigma}_n^{-1} \boldsymbol{\alpha}_n
\end{align}

Finally, there is
\begin{align}
    &\mathbb E^{\widehat{\mathbb{Q}}_n^+ \otimes \widehat{\mathbb{Q}}_n^-} \Big[ \| \boldsymbol{\Lambda}^* \boldsymbol{u} + \boldsymbol{\lambda}  \|_2^2  \Big] \nonumber \\
    =& \boldsymbol{\lambda}^T \boldsymbol{\lambda} - \boldsymbol{\lambda}^T \boldsymbol{\lambda}  \boldsymbol{\alpha}_n^T \boldsymbol{\Sigma}_n^{-1} \boldsymbol{\alpha}_n + \frac{1}{4n^2} \mathbb E^{\widehat{\mathbb{Q}}_n^+ \otimes \widehat{\mathbb{Q}}_n^-}\big[ \boldsymbol{u}^T\boldsymbol{\Sigma}_n^{-1}(\boldsymbol{\Sigma}_n - \boldsymbol{\Sigma}^*)\boldsymbol{\Sigma}_n^{-1}\boldsymbol{u} \big] 
\end{align}

As for the trace term, it becomes
\begin{align}
    \textbf{Tr}\big( \boldsymbol \Lambda (\boldsymbol{\Sigma}_n - \boldsymbol{\Sigma}^*) \big) = - \textbf{Tr} \big(\boldsymbol{\lambda} \boldsymbol{\alpha}_n^T \boldsymbol{\Sigma}_n^{-1} (\boldsymbol{\Sigma}_n - \boldsymbol{\Sigma}^*) - \frac{1}{2n} \textbf{Tr}\Big( (\boldsymbol{\Sigma}_n - \boldsymbol{\Sigma}^*) \boldsymbol{\Sigma}_n^{-1} (\boldsymbol{\Sigma}_n - \boldsymbol{\Sigma}^*) \Big)
\end{align}

Plug in the above formula, the Wasserstein robust profile becomes
\begin{align}
    &\mathcal{R}(\boldsymbol{\alpha}^*, \boldsymbol{\Sigma}^*) \nonumber \\
    =& \sup_{\boldsymbol{\lambda} \in \mathbb R^2} \Big\{  
        \boldsymbol{\lambda}^T (\boldsymbol{\alpha}^* - \boldsymbol{\alpha}_n) - n \| \boldsymbol{\lambda} \|^2 (1 - \boldsymbol{\alpha}_n^T \boldsymbol{\Sigma}_n^{-1} \boldsymbol{\alpha}_n) + \textbf{Tr} \big(\boldsymbol{\lambda} \boldsymbol{\alpha}_n^T \boldsymbol{\Sigma}_n^{-1} (\boldsymbol{\Sigma}_n - \boldsymbol{\Sigma}^*) \Big\} \nonumber \\
        +& \frac{1}{2n} \textbf{Tr}\Big( (\boldsymbol{\Sigma}_n - \boldsymbol{\Sigma}^*) \boldsymbol{\Sigma}_n^{-1} (\boldsymbol{\Sigma}_n - \boldsymbol{\Sigma}^*) \Big) + \frac{1}{4n} \mathbb E^{\widehat{\mathbb{Q}}_n^+ \otimes \widehat{\mathbb{Q}}_n^-} \Big[ \boldsymbol{u}^T \boldsymbol{\Sigma}_n^{-1}(\boldsymbol{\Sigma}_n - \boldsymbol{\Sigma}^*)^2 \boldsymbol{\Sigma}_n^{-1}\boldsymbol{u} \Big]
\end{align}

By taking the derivative of $\boldsymbol{\lambda}$, we get the following optimal $\boldsymbol{\lambda}^*$

Here, we show that $\boldsymbol{\alpha}_n^T \boldsymbol{\Sigma}_n^{-1} \boldsymbol{\alpha}_n < 1$. Since
\begin{align}
    \boldsymbol{\Sigma}_n = \begin{pmatrix}
        \beta_n^+ & \alpha_n^+ \alpha_n^- \\
        \alpha_n^+ \alpha_n^- & \beta_n^-
    \end{pmatrix}
\end{align}
then, 
\begin{align}
    \boldsymbol{\alpha}_n^T \boldsymbol{\Sigma}_n^{-1} \boldsymbol{\alpha}_n = \frac{ (\alpha_n^+)^2 \beta_n^- - 2(\alpha_n^+ \alpha_n^-)^2 + (\alpha_n^-)^2 \beta_n^+ }{ \beta_n^+ \beta_n^- - (\alpha_n^+ \alpha_n^-)^2 }
\end{align}
Since $\beta_n^\pm$ is the second moment, then we have $\beta_n^\pm = \text{Var}^{\widehat{\mathbb{Q}}_n^\pm}(u^\pm) + (\alpha_n^\pm)^2$, where $\boldsymbol{u} = (u^+, u^-)$. Then the above expression becomes
\begin{align}
        \boldsymbol{\alpha}_n^T \boldsymbol{\Sigma}_n^{-1} \boldsymbol{\alpha}_n =& \frac{ (\alpha_n^+)^2 \beta_n^- - 2(\alpha_n^+ \alpha_n^-)^2 + (\alpha_n^-)^2 \beta_n^+ }{ \beta_n^+ \beta_n^- - (\alpha_n^+ \alpha_n^-)^2 } \nonumber \\
        =& \frac{ (\alpha_n^+)^2 \text{Var}^{\widehat{\mathbb{Q}}_n^-}(u^-) + (\alpha_n^-)^2 \text{Var}^{\widehat{\mathbb{Q}}_n^+}(u^+) }{ \text{Var}^{\widehat{\mathbb{Q}}_n^+}(u^+)\text{Var}^{\widehat{\mathbb{Q}}_n^-}(u^-) + (\alpha_n^+)^2 \text{Var}^{\widehat{\mathbb{Q}}_n^-}(u^-) + (\alpha_n^-)^2 \text{Var}^{\widehat{\mathbb{Q}}_n^+}(u^+) }
\end{align}
Since the variance is by nature to be positive, then we prove our claim

Then the above equation has optimal $\boldsymbol{\lambda}^*$
\begin{align}
    \boldsymbol{\lambda}^* = \frac{1}{2n \left( 1 - \boldsymbol{\alpha}_n^T \boldsymbol{\Sigma}_n^{-1} \boldsymbol{\alpha}_n \right)} \left( \boldsymbol{\alpha}^* - \boldsymbol{\alpha}_n + (\boldsymbol{\Sigma}_n - \boldsymbol{\Sigma}^*)\boldsymbol{\Sigma}_n^{-1} \boldsymbol{\alpha}_n \right)
\end{align}

Plug in the optimal $\boldsymbol{\lambda}^*$,  
\begin{align}
    &(\boldsymbol{\lambda}^*)^T (\boldsymbol{\alpha}^* - \boldsymbol{\alpha}_n) - n \| \boldsymbol{\lambda}^* \|^2 (1 - \boldsymbol{\alpha}_n^T \boldsymbol{\Sigma}_n^{-1} \boldsymbol{\alpha}_n) \nonumber \\
    =& \frac{ (\boldsymbol{\alpha}^* - \boldsymbol{\alpha}_n)^T (\boldsymbol{\alpha}^* - \boldsymbol{\alpha}_n) }{2n  \left( 1 - \boldsymbol{\alpha}_n^T \boldsymbol{\Sigma}_n^{-1} \boldsymbol{\alpha}_n \right) } - \frac{\boldsymbol{\alpha}_n^T \boldsymbol{\Sigma}_n^{-1} (\boldsymbol{\Sigma}_n - \boldsymbol{\Sigma}^*)^2 (\boldsymbol{\alpha}^* - \boldsymbol{\alpha}_n)   }{2n  \left( 1 - \boldsymbol{\alpha}_n^T \boldsymbol{\Sigma}_n^{-1} \boldsymbol{\alpha}_n \right)} \nonumber \\ 
    -& \frac{ \left(  \boldsymbol{\alpha}^{*, T} - \boldsymbol{\alpha}_n^T  + \boldsymbol{\alpha}_n^T \boldsymbol{\Sigma}_n^{-1} (\boldsymbol{\Sigma}_n - \boldsymbol{\Sigma}^*) \right) \left(  \boldsymbol{\alpha}^{*} - \boldsymbol{\alpha}_n  + (\boldsymbol{\Sigma}_n - \boldsymbol{\Sigma}^*) \boldsymbol{\Sigma}_n^{-1} \boldsymbol{\alpha}_n^{-1} \right)  }{4n  \left( 1 - \boldsymbol{\alpha}_n^T \boldsymbol{\Sigma}_n^{-1} \boldsymbol{\alpha}_n \right)} \nonumber \\
    =& -\frac{ (\boldsymbol{\alpha}^* - \boldsymbol{\alpha}_n)^T (\boldsymbol{\alpha}^* - \boldsymbol{\alpha}_n) }{4n  \left( 1 - \boldsymbol{\alpha}_n^T \boldsymbol{\Sigma}_n^{-1} \boldsymbol{\alpha}_n \right) } - \frac{ \boldsymbol{\alpha}_n^T \boldsymbol{\Sigma}_n^{-1}  (\boldsymbol{\Sigma}_n - \boldsymbol{\Sigma}^*)^2
        \boldsymbol{\Sigma}_n^{-1}  \boldsymbol{\alpha}_n }{4n  \left( 1 - \boldsymbol{\alpha}_n^T \boldsymbol{\Sigma}_n^{-1} \boldsymbol{\alpha}_n \right) }
\end{align}

Also, for the third term in the above equation, there is
\begin{align}
    &\textbf{Tr} \big(\boldsymbol{\lambda}^* \boldsymbol{\alpha}_n^T \boldsymbol{\Sigma}_n^{-1} (\boldsymbol{\Sigma}_n - \boldsymbol{\Sigma}^*) \nonumber \\
    =& \textbf{Tr} \left\{ \frac{1}{2n  \left( 1 - \boldsymbol{\alpha}_n^T \boldsymbol{\Sigma}_n^{-1} \boldsymbol{\alpha}_n \right)} \Big(   \boldsymbol{\alpha}^* - \boldsymbol{\alpha}_n + (\boldsymbol{\Sigma}_n - \boldsymbol{\Sigma}^*)\boldsymbol{\Sigma}_n^{-1} \boldsymbol{\alpha}_n  \Big) \boldsymbol{\alpha}_n^T \boldsymbol{\Sigma}_n^{-1} (\boldsymbol{\Sigma}_n - \boldsymbol{\Sigma}^*)  \right\} \nonumber \\
    =& \textbf{Tr} \left\{ \frac{1}{2n  \left( 1 - \boldsymbol{\alpha}_n^T \boldsymbol{\Sigma}_n^{-1} \boldsymbol{\alpha}_n \right)}  \boldsymbol{\alpha}_n^T \boldsymbol{\Sigma}_n^{-1} (\boldsymbol{\Sigma}_n - \boldsymbol{\Sigma}^*)    \Big(   \boldsymbol{\alpha}^* - \boldsymbol{\alpha}_n + (\boldsymbol{\Sigma}_n - \boldsymbol{\Sigma}^*)\boldsymbol{\Sigma}_n^{-1} \boldsymbol{\alpha}_n  \Big) \right\} \nonumber \\
    =& \frac{\boldsymbol{\alpha}_n^T \boldsymbol{\Sigma}_n^{-1} (\boldsymbol{\Sigma}_n - \boldsymbol{\Sigma}^*)  (   \boldsymbol{\alpha}^* - \boldsymbol{\alpha}_n )}{2n  \left( 1 - \boldsymbol{\alpha}_n^T \boldsymbol{\Sigma}_n^{-1} \boldsymbol{\alpha}_n \right)} + \frac{  \boldsymbol{\alpha}_n^T \boldsymbol{\Sigma}_n^{-1} (\boldsymbol{\Sigma}_n - \boldsymbol{\Sigma}^*)^2  \boldsymbol{\Sigma}_n^{-1} \boldsymbol{\alpha}_n  }{2n  \left( 1 - \boldsymbol{\alpha}_n^T \boldsymbol{\Sigma}_n^{-1} \boldsymbol{\alpha}_n \right)} \nonumber
\end{align}

Then, the robust profile becomes 
\begin{align}
     &\mathcal{R}(\boldsymbol{\alpha}^*, \boldsymbol{\Sigma}^*) \nonumber \\
    =& \frac{ (\boldsymbol{\alpha}^* - \boldsymbol{\alpha}_n)^T (\boldsymbol{\alpha}^* - \boldsymbol{\alpha}_n) }{4n  \left( 1 - \boldsymbol{\alpha}_n^T \boldsymbol{\Sigma}_n^{-1} \boldsymbol{\alpha}_n \right) } + \frac{\boldsymbol{\alpha}_n \boldsymbol{\Sigma}_n^{-1} (\boldsymbol{\Sigma}_n - \boldsymbol{\Sigma}^*)^2 \boldsymbol{\Sigma}_n^{-1} \boldsymbol{\alpha}_n   }{4n  \left( 1 - \boldsymbol{\alpha}_n^T \boldsymbol{\Sigma}_n^{-1} \boldsymbol{\alpha}_n \right)} + \frac{\boldsymbol{\alpha}_n^T \boldsymbol{\Sigma}_n^{-1} (\boldsymbol{\Sigma}_n - \boldsymbol{\Sigma}^*)  (   \boldsymbol{\alpha}^* - \boldsymbol{\alpha}_n )}{2n  \left( 1 - \boldsymbol{\alpha}_n^T \boldsymbol{\Sigma}_n^{-1} \boldsymbol{\alpha}_n \right)} \nonumber \\ 
     +& \frac{1}{2n} \textbf{Tr}\Big( (\boldsymbol{\Sigma}_n - \boldsymbol{\Sigma}^*) \boldsymbol{\Sigma}_n^{-1} (\boldsymbol{\Sigma}_n - \boldsymbol{\Sigma}^*) \Big) + \frac{1}{4n} \mathbb E^{\widehat{\mathbb{Q}}_n^+ \otimes \widehat{\mathbb{Q}}_n^-} \Big[ \boldsymbol{u}^T \boldsymbol{\Sigma}_n^{-1}(\boldsymbol{\Sigma}_n - \boldsymbol{\Sigma}^*)^2 \boldsymbol{\Sigma}_n^{-1}\boldsymbol{u} \Big]
\end{align}
\end{proof}

\end{appendices}


\newpage


\begin{thebibliography}{99}

\bibitem{ho1981optimal} Ho, Thomas and Stoll, Hans R.\newblock Optimal dealer pricing under transactions and return uncertainty.\newblock \emph{Journal of Financial Economics}, 9(1):47--73, 1981.

\bibitem{avellaneda2008high} Avellaneda, Marco and Stoikov, Sasha.\newblock High-frequency trading in a limit order book.\newblock \emph{Quantitative Finance}, 8(3):217--224, 2008.

\bibitem{mai2021robust} Mai, Tien and Jaillet, Patrick.\newblock Robust Entropy-regularized Markov Decision Processes.\newblock \emph{arXiv preprint arXiv:2112.15364}, 2021.

\bibitem{stoikov2009option} Stoikov, Sasha and Sa{\u{g}}lam, Mehmet.\newblock Option market making under inventory risk.\newblock \emph{Review of Derivatives Research}, 12:55--79, 2009.

\bibitem{baldacci2021algorithmic} Baldacci, Bastien and Bergault, Philippe and Gu{\'e}ant, Olivier.\newblock Algorithmic market making for options.\newblock \emph{Quantitative Finance}, 21(1):85--97, 2021.

\bibitem{barzykin2021market} Barzykin, Alexander and Bergault, Philippe and Gu{\'e}ant, Olivier.\newblock Market making by an FX dealer: tiers, pricing ladders and hedging rates for optimal risk control.\newblock \emph{arXiv preprint arXiv:2112.02269}, 2021.

\bibitem{cartea2023decentralised} Cartea, {\'A}lvaro and Drissi, Fay{\c{c}}al and Monga, Marcello.\newblock Decentralised finance and automated market making: Predictable loss and optimal liquidity provision.\newblock \emph{arXiv preprint arXiv:2309.08431}, 2023.

\bibitem{blanchet2022distributionally} Blanchet, Jose and Chen, Lin and Zhou, Xun Yu.\newblock Distributionally robust mean-variance portfolio selection with Wasserstein distances.\newblock \emph{Management Science}, 68(9):6382--6410, 2022.

\bibitem{cartea2020market} Cartea, {\'A}lvaro and Wang, Yixuan.\newblock Market making with alpha signals.\newblock \emph{International Journal of Theoretical and Applied Finance}, 23(03):2050016, 2020.

\bibitem{cartea2019market} Cartea, {\'A}lvaro and Wang, Yixuan.\newblock Market making with minimum resting times.\newblock \emph{Quantitative Finance}, 19(6):903--920, 2019.

\bibitem{bergault2018closed} Bergault, Philippe and Evangelista, David and Gu{\'e}ant, Olivier and Vieira, Douglas.\newblock Closed-form approximations in multi-asset market making.\newblock \emph{arXiv preprint arXiv:1810.04383}, 2018.

\bibitem{blanchet2019robust} Blanchet, Jose and Kang, Yang and Murthy, Karthyek.\newblock Robust Wasserstein profile inference and applications to machine learning.\newblock \emph{Journal of Applied Probability}, 56(3):830--857, 2019.

\bibitem{gao2022wasserstein} Gao, Rui and Chen, Xi and Kleywegt, Anton J.\newblock Wasserstein distributionally robust optimization and variation regularization.\newblock \emph{Operations Research}, 2022.

\bibitem{yang2022wasserstein} Yang, Zhen and Gao, Rui.\newblock Wasserstein Regularization for 0-1 Loss.\newblock \emph{Optimization Online}, 2022.

\bibitem{blanchet2021statistical} Blanchet, Jose and Murthy, Karthyek and Nguyen, Viet Anh.\newblock Statistical analysis of Wasserstein distributionally robust estimators.\newblock In \emph{Tutorials in Operations Research: Emerging Optimization Methods and Modeling Techniques with Applications}, pages 227--254, 2021.

\bibitem{kuhn2019wasserstein} Kuhn, Daniel and Esfahani, Peyman Mohajerin and Nguyen, Viet Anh and Shafieezadeh-Abadeh, Soroosh.\newblock Wasserstein distributionally robust optimization: Theory and applications in machine learning.\newblock In \emph{Operations Research \& Management Science in the Age of Analytics}, pages 130--166, 2019.

\bibitem{blanchet2022confidence} Blanchet, Jose and Murthy, Karthyek and Si, Nian.\newblock Confidence regions in Wasserstein distributionally robust estimation.\newblock \emph{Biometrika}, 109(2):295--315, 2022.

\bibitem{zhang2022simple} Zhang, Luhao and Yang, Jincheng and Gao, Rui.\newblock A simple and general duality proof for Wasserstein distributionally robust optimization.\newblock \emph{arXiv preprint arXiv:2205.00362}, 2022.

\bibitem{mohajerin2018data} Mohajerin Esfahani, Peyman and Kuhn, Daniel.\newblock Data-driven distributionally robust optimization using the Wasserstein metric: Performance guarantees and tractable reformulations.\newblock \emph{Mathematical Programming}, 171(1):115--166, 2018.

\bibitem{gao2023distributionally} Gao, Rui and Kleywegt, Anton.\newblock Distributionally robust stochastic optimization with Wasserstein distance.\newblock \emph{Mathematics of Operations Research}, 48(2):603--655, 2023.

\end{thebibliography}
\end{document}